\documentclass{article}
\usepackage{palatino,epsfig,latexsym,natbib}

\title{Fitness Probability Distribution of Bit-Flip Mutation}

\author{Francisco Chicano$^1$, Andrew M. Sutton$^2$, \\L. Darrell Whitley$^2$ and Enrique Alba$^1$ \\
\\
$^1$University of M\'alaga, M\'alaga, Spain\\
$^2$Colorado State University, CO, USA
}

\date{}

\usepackage{amssymb}
\usepackage{amsmath}
\usepackage{amsthm}

\usepackage{array}
\usepackage{subfigure}
\usepackage{latexsym}
\usepackage{algorithm}
\usepackage{algorithmic}

\usepackage{graphicx}

\newlength{\marg}
\newlength{\margv}
\newtheorem{theorem}{Theorem}
\newtheorem{lemma}{Lemma}
\newtheorem{definition}{Definition}
\newtheorem{corollary}{Corollary}
\newtheorem{proposition}{Proposition}

\newcommand{\avg}[1]{\mathop{\mathrm{avg} \{f(y)\}}_{y \in #1}}

\def\vec#1{\mathchoice{\mbox{\boldmath$\displaystyle#1$}}
  {\mbox{\boldmath$\textstyle#1$}}
  {\mbox{\boldmath$\scriptstyle#1$}}
  {\mbox{\boldmath$\scriptscriptstyle#1$}}}

\newcommand{\comb}[2]{\binom{#1}{#2}}

\newcommand{\Bo}[0]{\mathbb{B}}

\newcommand{\Za}[0]{\mathbb{Z}}
\newcommand{\Real}[0]{\mathbb{R}}
\newcommand{\Exp}[1]{E\{#1\}}

\newcommand{\NP}[0]{\mathsf{NP}}
\renewcommand{\P}[0]{\mathsf{P}}

\newcommand{\krawel}[3]{\mathcal{K}^{#1}_{#2,#3}}
\newcommand{\krawelt}[3]{\tilde{\mathcal{K}}^{#1}_{#2,#3}}
\newcommand{\kraw}[1]{\mathcal{K}^{#1}}
\newcommand{\krawt}[1]{\tilde{\mathcal{K}}^{#1}}

\newcommand{\Prob}[1]{\mathop{\mathrm{Pr}\{#1\}}}
\newcommand{\ProbCond}[2]{\mathop{\mathrm{Pr}\{#1|#2\}}}
\newcommand{\Diag}[1]{\mathop{\mathrm{diag}\left(#1\right)}}

\newcommand{\dotproduct}[2]{\left\langle #1,#2 \right\rangle}

\newif\ifshowcomments

\showcommentstrue 

\begin{document}

\maketitle

\begin{abstract}
Bit-flip mutation is a common mutation operator for evolutionary algorithms applied to optimize functions over binary strings. In this paper, we develop results from the theory of landscapes and Krawtchouk polynomials to exactly compute the probability distribution of fitness values of a binary string undergoing uniform bit-flip mutation. We prove that this probability distribution can be expressed as a polynomial in $p$, the probability of flipping each bit. We analyze these polynomials and provide closed-form expressions for an easy linear problem (Onemax), and an $\NP$-hard problem, MAX-SAT. We also discuss some implications of the results for runtime analysis.
\end{abstract}

\parskip=0.00in

\section{Introduction}
\label{sec:intro}

Evolutionary algorithms that operate on binary string representations commonly employ the \emph{bit-flip} mutation operator.  This operator acts independently on each bit in a solution and changes the value of the bit (from 0 to 1 and vice versa) with probability $p$, where $p$ is a parameter of the operator.  The most commonly recommended value for this parameter is $p=1/n$ where $n$ is the length of the binary string. For linear functions, this rate is provably optimal~\citep{Witt2013tight}. However, in the general case, very little is currently understood about the mutation operator and its influence on the optimization process.

In this paper we study the operator from the point of view of landscape theory. Using this approach, we provide closed-form formulas for the fitness probability distribution of the solutions obtained after the application of bit-flip mutation to a particular solution. Up to the best of our knowledge, these kind of general and closed-form formulas have not been presented before. We can find, however, some works in which mathematical expressions are provided for this probability distribution in the case of particular problems like Onemax~\citep{Garnier1999}. In this paper we want to be more general and provide a mathematical expression for the probability distribution that separate two elements: the mathematical entity related to the problem, $\vec{F}$ and another one related to the operator $\vec{\Lambda}$. This approach yields a general framework that provides an expression that is valid for any problem as far as we can provide the problem-dependent entity $\vec{F}$.





\cite{Sutton2011gecco} and \cite{Chicano2011gecco}, used landscape theory to provide a closed-form formula for the expectation after a bit-flip mutation (we repeat this result in Section~\ref{subsec:expectation}). In this work we generalize these results and the one by \cite{Sutton2011foga}, providing closed-form formulas to compute the fitness probability distribution. The new result provides a deeper understanding of the behavior of the mutation operator. We illustrate the approach by providing concrete expressions for the moments of the probability distribution for two well-known problems: Onemax and MAX-SAT. 

Many results coming from the field of fitness landscape analysis provide exact results for the expected fitness of solutions undergoing uniform transformations by evolutionary operators. However, they often cannot say anything about selection operators because the framework does not easily handle the probability of obtaining an improving mutation. This quantity is much harder to derive because it always depends on some instance-dependent structure that is ignored in such analyses. In this paper, we work out how to address the problem of selection by separating the instance-dependent structure from the instance-independent structure. This separation allows us to derive initial results on the probability of producing an improving offspring. As a consequence, we illustrate a way to use the theory of landscapes to derive the expected runtime of a $(1+\lambda)$ EA without crossover on Onemax.

The remainder of the paper is organized as follows. In the next section the mathematical tools required to understand the rest of the paper are presented. In Section~\ref{sec:mutation} we present our main contribution of this work: the landscape analysis of bit-flip mutation and the closed-form formulas for the fitness probability distribution. Section~\ref{sec:study} provides particular results for two well-known problems in the domain of combinatorial optimization: Onemax and MAX-SAT. Section~\ref{sec:applications} proves the connection between the results in this paper and the runtime analysis of a $(1+\lambda)$ EA. Finally, Section~\ref{sec:conclusions} presents the conclusions and future work.

\section{Background}
\label{sec:background}

In this section we present some fundamental results of landscape theory. We will only focus on the relevant information required to understand the rest of the paper. We refer the reader interested in a deeper exposition of this topic to the survey by \cite{Reidys2002}.

A \emph{landscape} for a combinatorial optimization problem is a triple $(X,N,f)$, where $X$ is a finite or countable solution set, $f : X \rightarrow \mathbb{R}$ defines the objective function and $N$ is a \emph{neighborhood function} that maps any solution $x \in X$ to the set $N(x)$ of points reachable from $x$. If $y \in N(x)$ then $y$ is a neighbor of $x$.
The pair $(X,N)$ is called \emph{configuration space} and can be represented using a graph $G = (X,E)$ in which $X$ is the set of vertices and a directed edge $(x,y)$ exists in $E$ if $y \in N(x)$ \citep{Biyikoglu2007}.  We can represent the neighborhood operator by its adjacency matrix
\begin{equation}
\vec{A}_{x,y} = \left\{
\begin{array}{ll}
1 & \mbox{if $y \in N (x)$}, \\
0 & \mbox{otherwise.}
\end{array}
\right.
\end{equation}

Any discrete function, $f$, defined over the set of candidate solutions can be characterized as a vector in $\mathbb{R}^{|X|}$. Any $|X| \times |X|$ matrix can be interpreted as a linear map that acts on vectors in $\mathbb{R}^{|X|}$. For example, the adjacency matrix $\vec{A}$ acts on function $f$ as follows
\begin{equation}
\vec{A} ~ f = \left(
\begin{array}{c}
\sum_{y \in N(x_1)} f(y) \\
\sum_{y \in N(x_2)} f(y) \\
\vdots \\
\sum_{y \in N(x_{|X|})} f(y)
\end{array}
\right).
\end{equation}

The component $x$ of this matrix-vector product can thus be written
as:
\begin{equation}
\label{eqn:matrix-vector} (\vec{A} ~ f) (x) = \sum_{y \in N(x)} f(y),
\end{equation}
which is the sum of the function value of all the neighbors of $x$. In the case of binary strings, the minimal-change neighborhood at a point $x$ is the set of Hamming neighbors of $x$. The Hamming neighborhood induces a regular, connected graph $G =(X,E)$, meaning that $G$ is connected and $|N (x)| = d > 0$ for a constant $d$, for all $x \in X$. When a neighborhood is regular, the so-called \emph{Laplacian matrix} is defined as $\vec{\Delta} = \vec{A} - d \vec{I}$. This corresponds to the Laplacian of the graph $G$. Stadler defines the class of \emph{elementary landscapes} where the function $f$ is an eigenvector (or eigenfunction) of the Laplacian up to an additive constant~\citep{Stadler1995landscapes}. Formally, we have the following.

\begin{definition}
Let $(X, N, f)$ be a landscape and $\vec{\Delta}$ be the Laplacian matrix of the configuration space graph. The landscape is said to be elementary if there exists a constant $b$ that we call the \emph{offset}, and an eigenvalue $\lambda$ of $-\vec{\Delta}$ such that $(-\vec{\Delta}) (f-b) = \lambda (f-b)$. 
\end{definition}

We use eigenvalues of $-\vec{\Delta}$ instead of $\vec{\Delta}$ to have positive eigenvalues~\citep{Biyikoglu2007}. In connected neighborhoods like the Hamming neighborhood, the offset $b$ is the average value of the function $f$ evaluated over the entire search space: $b=\bar{f}$. In elementary landscapes, the average value $\bar{f}$ can be usually computed in a very efficient way using the problem data. That is, it is not required to do a complete enumeration over the search space. For a concrete example on the TSP the reader is referred to~\cite{Whitley2008}.

Suppose $(X,N,f)$ is elementary with eigenvalue $\lambda$. For any scalars $a$ and $b$, define the function $g : X \to \Real$ as $g(x) = af(x) + b$. Clearly, $(X,N,g)$ is also elementary with the same eigenvalue $\lambda$. Furthermore, in regular neighborhoods, if $g$ is an eigenfunction of $-\vec{\Delta}$ with eigenvalue $\lambda$ then $g$ is also an eigenfunction of $\vec{A}$ (the adjacency matrix of the configuration space graph $G$) with eigenvalue $d-\lambda$. The average value of the fitness function in the neighborhood of a solution can be computed using the expression:

\begin{equation}
\label{eqn:avg} \avg{N(x)} = \frac{1}{d} (\vec{A} ~f) (x).
\end{equation}
If $(X,N,f)$ is elementary with eigenvalue $\lambda$, then the average over the neighborhood is computed as:
\begin{align}
\nonumber \avg{N(x)} &= \mathop{\mathrm{avg}}_{y \in N(x)} \{f(y)-\bar{f}\} + \bar{f} \\
\nonumber &= \frac{1}{d} (\vec{A} ~(f-\bar{f})) (x) + \bar{f}= \frac{d-\lambda}{d} (f(x)-\bar{f}) + \bar{f} \\
&= f(x) + \frac{\lambda}{d} (\bar{f} -f(x)),
\end{align}
which is sometimes referred to as Grover's wave equation~\citep{Grover1992local}. In the previous expression we used the fact that $f-\bar{f}$ is an eigenfunction of $\vec{A}$ with eigenvalue
$d-\lambda$.

The wave equation makes it possible to compute the average value of the fitness function $f$ evaluated over all of the neighbors of $x$ using only the value $f(x)$.
The previous average can be interpreted as the expected value of the objective function when a random neighbor of $x$ is selected using a uniform distribution. This is exactly the behavior of the so-called \emph{1-bit-flip} mutation~\citep{Garnier1999}.

A landscape $(X,N,f)$ is not always elementary, but even in this case it is possible to characterize the function $f$ as the sum of elementary landscapes, called \emph{elementary components} of the landscape. The interested reader can find examples of elementary landscapes in~\cite{Whitley2008,WhitleySutton2009} and can find more on the elementary landscape decomposition in~\cite{Chicano2011ecj}.

\subsection{Binary Hypercube}
\label{subsec:binary}

The previous definitions are general concepts of landscape theory. Let us focus now on the binary configuration spaces with the Hamming neighborhood, the so-called \emph{binary hypercubes}, which are the configuration spaces we need in the analysis of bit-flip mutation. Let us first present the notation. In these spaces the solution set $X$ is the set of all binary strings of size $n$, formally, $\Za_2^n=\Bo^n$. The solution set form an Abelian group with the component-wise sum in $\Za_2$ (exclusive OR), denoted with $\oplus$. Given an element $z \in \Bo^n$, we will denote with $|z|$ the number of ones of $z$. Given a set of binary strings $W$ and a binary string $u$ we denote with $W \wedge u$ the set of binary strings that can be computed as the bitwise AND of a string in $W$ and $u$, that is, $W \wedge u = \{w \wedge u | w \in W\}$. For example, $\Bo^4 \wedge 0101 = \{0000, 0001, 0100, 0101\}$. We will denote with $\underline{i}$ the binary string with position $i$ set to 1 (starting from the leftmost position) and the rest set to 0. We omit the length of the string $n$ in the notation, but it will be clear from the context. For example, if we are considering binary strings in $\Bo^4$ we have $\underline{1}=1000$ and $\underline{3}=0010$.

It is convenient to characterize the neighborhood by a set of group elements $S=\{s_1,s_2,\ldots,s_d\}$ that generate the entire group. Here $S$ is called a \emph{generating set}. The neighborhood of a solution $x$ is just the set $N(x) = x \oplus S = \{x \oplus s | s \in S\}$. In the binary hypercube, two solutions $x$ and $y$ are neighbors if one can be obtained from the other by flipping a single bit, that is, if the Hamming distance between the solutions, $|x\oplus y|$, is 1. Thus, the generating set is composed of every binary string with a single 1: $S_1=\{s \in \Bo^n \mid |s|=1\}$.

We define the sphere of radius $r$  around a solution $x$ as the set of all solutions lying at Hamming distance $r$ from $x$~\citep{Sutton2010}. We are also interested in these spheres since the probability of reaching a solution $y$ from a solution $x$ using the bit-flip mutation operator is the same for all the solutions in a sphere around $x$. Now we can observe that the solutions in a sphere of radius $r$ around $x$ can be thought as the neighborhood $N_r$ of $x$ generated by an appropriate generating set $S_r$. The generating set is composed of all the solutions having exactly $r$ 1s: $S_r=\{s \in \Bo^n \mid |s|=r\}$. The notation $S_1$ used before was selected to be a particular case of this more general neighborhood. We will use the notation $N_r(x)=x \oplus S_r$. Another particular case is the one of $S_0=\{0\}$ that generates the identity neighborhood $N_0(x)=\{x\}$. Each neighborhood has its corresponding adjacency matrix denoted with $\vec{A}^{(r)}$.


Let us consider the set of all the pseudo-Boolean functions defined over $\Bo^n$, $\Real^{\Bo^n}$. We can think of one pseudo-Boolean function as an array of $2^n$ real numbers, each one being the function evaluation of a particular binary string of $\Bo^n$. Each pseudo-Boolean function is, thus, a particular vector in a vector space with $2^n$ dimensions. Let us define the dot-product between two pseudo-Boolean functions as:
\begin{equation}
\dotproduct{f}{g} = \sum_{x \in \Bo^n} f(x) g(x).
\end{equation}

Now we introduce a set of functions that will be relevant for our purposes in the next sections: the \emph{Walsh functions}~\citep{Walsh1923}

\begin{definition}
The (non-normalized) Walsh function with parameter $w \in \Bo^n$ is a pseudo-Boolean function defined over $\Bo^n$ as:
\begin{equation}
\label{eqn:walsh-def}
\psi_{w}(x) = \prod_{i=1}^{n} (-1)^{w_i x_i} = (-1)^{\sum_{i=1}^n w_i x_i},
\end{equation}
where the subindex in $w_i$ and $x_i$ denotes the $i$-th component of the binary strings $w$ and $x$, respectively.
\end{definition}

We can observe that the Walsh functions map $\Bo^n$ to the set $\{-1,1\}$. We define the \emph{order} of a Walsh function $\psi_w$ as the value $|w|$. Some properties of the Walsh functions are given in the following proposition. A proof of these properties can be found in~\cite{Vose1999}.

\begin{proposition}
Let us consider the Walsh functions defined over $\Bo^n$. The following identities hold:
\begin{align}
\label{eqn:walsh-zero}\psi_{0} &= 1, \\
\label{eqn:walsh-oplus} \psi_{w \oplus t} &= \psi_w  \psi_{t} ,\\
\label{eqn:sum-arg}\psi_{w}(x \oplus y) &= \psi_{w}(x) \psi_{w}(y), \\
\label{eqn:walsh-swap}\psi_{w}(x) &= \psi_{x}(w), \\
\label{eqn:walsh-square}\psi_{w}^2 &= 1, \\
\label{eqn:walsh-sum} \sum_{x \in \Bo^n} \psi_{w}(x) &= 2^n \delta_0^{|w|} 
= \left\{ \begin{array}{ll} 2^n & \mbox{if $w=0$,} \\ 0 & \mbox{if $w \neq 0$,}\end{array}\right.\\
\label{eqn:walsh-bit}\psi_{\underline{i}}(x) &= (-1)^{x_i} = 1 -2 x_i, \\
\label{eqn:walsh-dot}\dotproduct{\psi_{w}}{\psi_{t}} &= 2^n \delta_{w}^t,
\end{align}
where $\delta$ denotes the Kronecker delta.
\end{proposition}

There exist $2^n$ Walsh functions in $\Bo^n$ and according to (\ref{eqn:walsh-dot}) they are orthogonal, so they form a basis of the set of pseudo-Boolean functions. Any arbitrary pseudo-Boolean function $f$ can be expressed as a weighted sum of Walsh functions. We can represent $f$ in the Walsh basis in the following way:
\begin{equation}
f(x) = \sum_{w \in \Bo^n} a_w \psi_{w}(x),
\end{equation}
where the \emph{Walsh coefficients} $a_w$ are defined as:
\begin{equation}
\label{eqn:walsh-coeff}
a_w = \frac{1}{2^n} \dotproduct{\psi_w}{f}.
\end{equation}
The previous expression is called the \emph{Walsh expansion} (or decomposition) of $f$. The interested reader can refer to the text by \citet{Terras1999} for a deeper treatment of Walsh functions and their properties.

The reason why Walsh functions are so important for the mutation analysis is because they are eigenvectors of the adjacency matrices $\vec{A}^{(r)}$ defined above, as the next proposition proves.

\begin{proposition}
\label{prop:eigenvalue-walsh}
In $\Bo^n$, the Walsh function $\psi_w$ defined in (\ref{eqn:walsh-def}) is an eigenvector of the adjacency matrix $\vec{A}^{(r)}$ based on the generating set $S_r$ (sphere of radius $r$) with eigenvalue
\begin{equation}
\label{eqn:eigenvalue-walsh}
\sum_{s \in S_r} \psi_w (s) = \krawel{(n)}{r}{|w|},
\end{equation}
where $\krawel{(n)}{r}{j}$ is the $(r,j)$ element of the so-called $n$-th order Krawtchouk matrix $\kraw{(n)}$, defined as:
\begin{equation}
\label{eqn:kra-def}
\krawel{(n)}{r}{j} = \sum_{l=0}^{n} (-1)^ l \comb{n-j}{r-l}\comb{j}{l},
\end{equation}
for $0\leq r,j \leq n$. We assume in the previous expression that $\comb{a}{b} =0$ if $b > a$ or $b < 0$. 
\end{proposition}
\begin{proof}
The Walsh function $\psi_w$ is an eigenvector of $\vec{A}^{(r)}$ if $\vec{A}^{(r)} \psi_w = \lambda \psi_w$ for some constant $\lambda$, which is the eigenvalue. Taking into account the definition of neighborhood based on the generating set $S_r$ we can write:
\begin{align}
\left(\vec{A}^{(r)} \psi_w\right)(x) 
&= \sum_{s \in S_r} \psi_w (x\oplus s) 
= \sum_{s \in S_r} \psi_w(x) \psi_w (s) 
= \left( \sum_{s \in S_r} \psi_w(s) \right)\psi_w(x), \nonumber
\end{align}
where we used the property (\ref{eqn:sum-arg}) and we can identify the eigenvalue with the left hand side of (\ref{eqn:eigenvalue-walsh}). Let us now prove that this value is exactly $\krawel{(n)}{r}{|w|}$. Using the definition of $S_r$ we can write the series as:
\begin{align}
\sum_{s \in S_r} \psi_w(s) &= \sum_{s \in \Bo^n \atop |s|=r} \psi_w(s) = \sum_{s \in \Bo^n \atop |s|=r} (-1)^{|w \wedge s|},
\end{align}
and we can now change the index of the sum from $s$ to $l=|w \wedge s|$. Written with the new index we only need to count for each $l$ how many binary strings $s \in S_r$ have the property that $|w \wedge s|=l$, that is:
\begin{align}
\label{eqn:card}
\sum_{s \in \Bo^n \atop |s|=r} (-1)^{|w \wedge s|} = \sum_{l=0}^n (-1)^l \left| \{s \in \Bo^n| |s|=r \text{ and } |w \wedge s|=l\} \right|.
\end{align}

Now we can compute the cardinality of the inner set in (\ref{eqn:card}) using counting arguments. We need to count how many ways we can distribute the $r$ 1s in the string $s$ such that they coincide with the 1s of $w$ in exactly $l$ positions. In order to do this, first let us put $l$ 1s in the positions where $w$ has 1. We can do this in $\comb{|w|}{l} $ different ways. Now, let us put the remaining $r-l$ 1s in the positions where $w$ has 0. We can do this in $\comb{n-|w|}{r-l}$ ways. Multiplying both numbers we have the desired cardinality:
\begin{equation}
\label{eqn:card2}
\left| \{s \in \Bo^n| |s|=r \text{ and } |w \wedge s|=l\} \right| = \comb{|w|}{l} \comb{n-|w|}{r-l}.
\end{equation}

We should notice here that the cardinality is zero in some cases. This happens when $l > |w|$, $l>r$ or $r-l > n-|w|$. However, in these cases we defined the binomial coefficient to be zero and we can keep the previous expression. If we use (\ref{eqn:card2}) in (\ref{eqn:card}) and take into account the definition (\ref{eqn:kra-def}) we get (\ref{eqn:eigenvalue-walsh}).
\end{proof}

In (\ref{eqn:eigenvalue-walsh}) we can observe that the eigenvalue depends only on the order $|w|$ of the Walsh function. This means that there are at most $n+1$ different eigenvalues in the considered adjacency matrices. As a consequence, we can decompose any arbitrary function $f$ as a sum of $n+1$ functions, called \emph{elementary components} of $f$, where each one is an eigenvector of all the adjacency matrices. 

\begin{definition}
Let $f:\Bo^n \rightarrow \Real$ be a pseudo-Boolean function with Walsh expansion $f= \sum_{w \in \Bo^n} a_w \psi_w$, we define the order-$j$ elementary component of $f$ as:
\begin{equation}
f_{[j]} = \sum_{w \in \Bo^n \atop |w|=j} a_w \psi_w,
\end{equation}
for $0\leq j \leq n$. As a consequence of the Walsh expansion of $f$ we can write:
\begin{equation}
f = \sum_{j=0}^{n} f_{[j]} .
\end{equation}
\end{definition}

According to Proposition~\ref{prop:eigenvalue-walsh} the elementary component $f_{[j]}$ is an eigenvector of $\vec{A}^{(r)}$ with eigenvalue $\krawel{(n)}{r}{j}$.

\subsection{Krawtchouk Matrices}

Krawtchouk matrices play a relevant role in the mathematical developments of the next sections. For this reason we present here some of their properties. The reader interested in these matrices (also considered polynomials) can read~\cite{Feinsilver2005}. The $n$-th order Krawtchouk matrix is an $(n+1)\times(n+1)$ integer matrix with indices between $0$ and $n$. In (\ref{eqn:kra-def}) we provided an explicit definition of the elements of a Krawtchouk matrix. But these elements can also be implicitly defined with the help of the following generating function:
\begin{equation}
\label{eqn:kr-generating}
(1+x)^{n-j}(1-x)^{j} = \sum_{r=0}^n x^r \krawel{(n)}{r}{j}.
\end{equation}

From (\ref{eqn:kr-generating}) we deduce that $\krawel{(n)}{0}{j} = 1$. Observe that $\krawel{(n)}{0}{j}$ is the constant coefficient in the polynomial. Other properties of the Krawtchouk matrices are presented in the next proposition.
\begin{proposition}
\label{prop:kraw-col}
We have the following identities between the elements of the Krawtchouk matrices:
\begin{equation}
\label{eqn:kraw-col}
\krawel{(n)}{r}{n-j} = (-1)^r  \krawel{(n)}{r}{j} ,
\end{equation}
\begin{equation}
\label{eqn:kraw-row}
\krawel{(n)}{n-r}{j} = (-1)^j  \krawel{(n)}{r}{j} .
\end{equation}

\end{proposition}
\begin{proof}
With the help of the generating function (\ref{eqn:kr-generating}) we can write:
\begin{align}
\sum_{r=0}^n (-x)^r \krawel{(n)}{r}{j} &= (1+(-x))^{n-j}(1-(-x))^{j} \nonumber \\
&= (1+x)^{j}(1-x)^{n-j} = \sum_{r=0}^n x^r \krawel{(n)}{r}{n-j}, \nonumber
\end{align}
and identifying the coefficients of the first and last polynomials we have (\ref{eqn:kraw-col}). In order to prove (\ref{eqn:kraw-row}) we can write:
\begin{align}
\sum_{r=0}^n x^r \krawel{(n)}{r}{j} &= (1+x)^{n-j}(1-x)^{j} = (-1)^j(x+1)^{n-j}(x-1)^{j} \nonumber \\
&= (-1)^j x^n (1+1/x)^{n-j}(1-1/x)^{j} = (-1)^j x^n \sum_{r=0}^n (1/x)^r \krawel{(n)}{r}{j} \nonumber \\
&= (-1)^j \sum_{r=0}^n x^{n-r} \krawel{(n)}{r}{j} = (-1)^j \sum_{r=0}^n x^{r} \krawel{(n)}{n-r}{j}, \nonumber
\end{align}
identifying again the coefficients of the first and last polynomials we have (\ref{eqn:kraw-row}).
\end{proof}

Krawtchouk matrices also appear when we sum Walsh functions. The following proposition provides an important result in this line.
\begin{proposition}
\label{prop:krawtchouk}
Let $t \in \Bo^n$ be a binary string and $0 \leq r \leq n$. Then the following two identities hold for the sum of Walsh functions:
\begin{align}
\label{eqn:sum-walsh-r}\sum_{w \in \Bo^n \wedge t \atop |w| = r} \psi_{w}(x) &= \krawel{(|t|)}{r}{|x \wedge t|} , \\
\label{eqn:sum-walsh}\sum_{w \in \Bo^n \wedge t} \psi_{w}(x) &= 2^{|t|} \delta_0^{|x \wedge t|} .
\end{align}
\end{proposition}
\begin{proof}
Given two binary strings $x, t \in \Bo^n$, let us denote with $x|_t$ the binary string of length $|t|$ composed of all the bits of $x$ in the positions $i$ where $t_i=1$. The string $t$ acts as a mask for $x$. This notation allows us to simplify the sums in (\ref{eqn:sum-walsh-r}) and (\ref{eqn:sum-walsh}):
\begin{align}
\sum_{w \in \Bo^n \wedge t \atop |w| = r} \psi_{w}(x) &= 
\sum_{w \in \Bo^n \wedge t \atop |w| = r} \psi_{w|_t}(x|_t) =
\sum_{u \in \Bo^{|t|} \atop |u| = r} \psi_{u}(x|_t) = \sum_{u \in S_r} \psi_{u}(x|_t) = \krawel{(|t|)}{r}{|x \wedge t|} ~~~~ \text{by (\ref{eqn:eigenvalue-walsh}),} \nonumber \\
\sum_{w \in \Bo^n \wedge t} \psi_{w}(x) &= 
\sum_{w \in \Bo^n \wedge t} \psi_{w|_t}(x|_t) =
\sum_{u \in \Bo^{|t|}} \psi_{u}(x|_t) = \sum_{u \in \Bo^{|t|}} \psi_{x|_t}(u) = 2^{|t|} \delta_0^{|x \wedge t|} ~~~~ \text{by (\ref{eqn:walsh-sum})} . \nonumber 
\end{align}

\end{proof}


\section{Analysis of the Mutation Operator}
\label{sec:mutation}

The bit-flip mutation operator transforms an arbitrary element $x \in \Bo^n$ to $y \in \Bo^n$ by changing the value of each bit of $x$ with probability $p$. In the literature it is common to use the value $p=1/n$ that, in expectation, changes one bit in each solution. However, if $0 < p < 1$, the mutation operator can transform $x$ into any element of the search space with positive probability. In the following, we denote with $M_p(x)$ the random variable on $\Bo^n$ that represents the element in $\Bo^n$ reached after applying the bit-flip mutation operator with probability $p$ to solution $x$.

\begin{lemma}
\label{lem:prob-mut}
Given two solutions $x,y \in \Bo^n$, the probability of obtaining $y$ after a bit-flip mutation over $x$ is
\begin{equation}
\label{eqn:prob-mut}
\Prob{M_p(x)=y} = p^{|x\oplus y|}(1-p)^{n-|x\oplus y|}.
\end{equation}
\end{lemma}
\begin{proof}
The solution $y$ can only be obtained if all the bits that differ from the solution $x$ are mutated and the other ones are kept unchanged. Since the number of differing bits is $|x\oplus y|$ and each bit is individually changed with probability $p$ we obtain the claimed result.
\end{proof}

We are interested in $f(M_p(x))$, the objective function value after the mutation of a solution. This value is also a random variable and we want to analyze its probability distribution. 
Given a particular search space, directly enumerating this distribution by evaluating every solution is not tractable. However, the theory of landscapes provides tools for extracting information from this probability distribution in an efficient way. This information arises from the moments of the probability distribution. In the following sections we analyze these moments.

\subsection{Expectation}
\label{subsec:expectation}

Let us start by computing the expected value of $f(M_p(x))$. The expected value is easy to compute in the case of the elementary components of a function $f$. The result of the next theorem was previously published by~\cite{Chicano2011gecco}. \cite{Sutton2011gecco} also studied the expected value after mutation and found that it must be a polynomial in $p$. We, however, present here the result and its proof because the notation is slightly different from the one used in the previous works.

\begin{theorem}
\label{thm:mut-exp-elementary}
Let $x \in \Bo^n$ be a binary string, $f:\Bo^n\to \Real$ a function, $f_{[j]}$ its order-$j$ elementary component and let us denote with $M_p(x)$ the random variable that represents the element in $\Bo^n$ reached after applying the bit-flip mutation operator with probability $p$ to solution $x$. The expected value of the random variable $f_{[j]}(M_p(x))$ is
\begin{equation}
\label{eqn:mut-exp-elementary}
\Exp{f_{[j]}(M_p(x))} = (1-2p)^j f_{[j]}(x)
\end{equation}
\end{theorem}
\begin{proof}
\begin{align}
\Exp{f_{[j]}(M_p(x))} &= \sum_{y \in \Bo^n} f_{[j]}(y) \Prob{M_p(x)=y} \notag \\
&= \sum_{y \in \Bo^n} f_{[j]}(y) p^{|x\oplus y|}(1-p)^{n-|x\oplus y|} & \text{by Lemma~\ref{lem:prob-mut}} \notag \\
&= \sum_{r=0}^{n} \sum_{y \in N_r(x)} f_{[j]}(y) p^{|x\oplus y|}(1-p)^{n-|x\oplus y|} & \text{dividing the search space} \notag \\
&= \sum_{r=0}^{n} \left( p^{r}(1-p)^{n-r} \sum_{y \in N_r(x)} f_{[j]}(y) \right) & \text{by definition of $N_r$} \notag \\
&= \sum_{r=0}^{n} p^{r}(1-p)^{n-r} \left( \vec{A}^{(r)} f_{[j]} \right)(x) & \text{by Proposition~\ref{prop:eigenvalue-walsh}} \notag \\
&= \left(\sum_{r=0}^{n} p^{r}(1-p)^{n-r} \krawel{(n)}{r}{j}\right) f_{[j]}(x) .
\end{align}

Using the generating function for Krawtchouk matrices (\ref{eqn:kr-generating}) we can simplify the term within the parentheses in the following way:

\begin{align}
\sum_{r=0}^{n} p^{r}(1-p)^{n-r} \krawel{(n)}{r}{j} &= (1-p)^{n} \sum_{r=0}^n \left(\frac{p}{1-p}\right)^r  \krawel{(n)}{r}{j} \notag \\
&=(1-p)^n   \left(1+\frac{p}{1-p}\right)^{n-j} \left(1-\frac{p}{1-p}\right)^j \notag \\
&=(1-p)^n   \left(\frac{1}{1-p}\right)^{n-j} \left(\frac{1-2p}{1-p}\right)^j \notag \\
&=(1-2p)^j.
\end{align}

The previous development is valid if $p < 1$. In the case $p=1$ we cannot divide by $1-p$, but even in this case the final result holds. To prove this we just have to consider that the term $(1-p)^{n-r}$ is zero except for $r=n$ and $p^r$ is always 1. Then we can write
\begin{equation}
\sum_{r=0}^{n} p^{r}(1-p)^{n-r} \krawel{(n)}{r}{j} = \krawel{(n)}{n}{j} = (-1)^j \krawel{(n)}{0}{j} = (1-2p)^j,
\end{equation}
where we used (\ref{eqn:kraw-row}) and the fact that $\krawel{(n)}{0}{j}=1$. We finally obtain the claimed result for all the possible values of $p$.
\end{proof}

As a direct consequence of the previous theorem we can compute the expected value of $f(M_p(x))$ for an arbitrary function with the help of the decomposition of the function into elementary components.

\begin{corollary}
Let $x \in \Bo^n$ be a binary string, $f:\Bo^n\to \Real$ a function and $M_p(x)$ the solution reached after applying the bit-flip mutation operator with probability $p$ to solution $x$. The expected value of the random variable $f(M_p(x))$ is
\begin{equation}
\label{eqn:expectation-decomp}
\Exp{f(M_p(x))} = \sum_{j=0}^n (1-2p)^j f_{[j]}(x) ,
\end{equation}
where $f_{[j]}$ is the order-$j$ elementary component of $f$.
\end{corollary}
\begin{proof}
We can write $f$ as the sum of its elementary components as $f=\sum_{j=0}^n f_{[j]}$. Then, we can compute the expected value as:
\begin{align}
\Exp{f(M_p(x))} &= \sum_{j=0}^n \Exp{f_{[j]}(M_p(x))} = \sum_{j=0}^n (1-2p)^j f_{[j]}(x),
\end{align}
where we used the result of Theorem~\ref{thm:mut-exp-elementary}.
\end{proof}

\subsection{Higher Order Moments}
\label{subsec:higher-moments}

Equation (\ref{eqn:expectation-decomp}) can be used to compute the expected value of $f(M_p(x))$. We may also use it to extend to higher order moments, as in the following theorem. 
\begin{theorem}
Let $x \in \Bo^n$ be a binary string, $f:\Bo^n\to \Real$ a function and $M_p(x)$ the solution reached after applying the bit-flip mutation operator with probability $p$ to solution $x$. The $m$-th moment of the random variable $f(M_p(x))$ is
\begin{equation}
\label{eqn:cth-order}
\mu_m\{f(M_p(x))\} = \sum_{j=0}^n (1-2p)^j f^m_{[j]}(x) ,
\end{equation}
where $f^m_{[j]}$ is the order-$j$ elementary component of $f^m$.\footnote{We use this notation instead of $(f^m)_{[j]}$ to simplify the expressions, but $f^m_{[j]}$ should not be confused with $f_{[j]}$ to the power of $m$.}
\end{theorem}
\begin{proof}
By definition, $\mu_m\{f(M_p(x))\}$ can be expressed as the expectation of the random variable $f^m(M_p(x))$. Then, using (\ref{eqn:expectation-decomp}) we can write:
\begin{align}
\mu_m\{f(M_p(x))\} = \Exp{f^m(M_p(x))} = \sum_{j=0}^n (1-2p)^j f^m_{[j]}(x) . \notag
\end{align}
\end{proof}

We define the $0$-th moment $\mu_0\{f(M_p(x))\}=1$. We can observe from (\ref{eqn:cth-order}) that all the higher-order moments are polynomials in $p$, just like the expectation (first order moment). 

Let us now introduce some new notation. Let us denote with $\vec{\mu}\{f(M_p(x))\}$ the vector of moments, that is, the $m$-th component of this vector is the $m$-th moment. We do not limit the number of components of this vector, we can consider it as an infinite-dimensional vector. Later we will see that only a finite number of elements of this vector would be required for our purposes. We define the matrix function $\vec{F}(x)$ as $\vec{F}_{m,j}(x) = f^m_{[j]}(x)$ where $0 \leq j \leq n$ and $m\geq0$. Let us also define the vector $\vec{\Lambda}(p)$ as $\vec{\Lambda}_j(p)=(1-2p)^j$ for $0 \leq j \leq n$. 

Using the new notation we can write (\ref{eqn:cth-order}) in vector form as:
\begin{equation}
\nonumber
\left( \begin{array}{c}\mu_0\{f(M_p(x))\} \\ \mu_1\{f(M_p(x))\} \\ \vdots \\  \mu_m\{f(M_p(x))\} \\ \vdots \end{array}\right) = \left( \begin{array}{cccc} 
1 & 0 & \ldots & 0 \\
f_{[0]}(x) & f_{[1]}(x) & \ldots & f_{[n]}(x) \\
\vdots & \vdots & \ddots & \vdots \\
f^m_{[0]}(x) & f^m_{[1]}(x) & \ldots & f^m_{[n]}(x) \\
\vdots & \vdots & \ddots & \vdots \\
\end{array}\right) \left( \begin{array}{c}1 \\ 1-2p \\ \vdots \\  (1-2p)^n  \end{array} \right),
\end{equation}
or in a compact way:
\begin{equation}
\label{eqn:muc-vector}
\vec{\mu}\{f(M_p(x))\} = \vec{F}(x) \vec{\Lambda}(p),
\end{equation}
where $\vec{F}(x)$ and $\vec{\Lambda}(p)$ are multiplied using the matrix product. This new form of writing (\ref{eqn:cth-order}) has the property of expressing the vector of moments of $f(M_p(x))$ as the product of a matrix that depends on the objective function (and solution $x$) and a vector that depends on the mutation operator and its parameter $p$. In some sense, we can claim that (\ref{eqn:muc-vector}) decomposes the moments in a problem-dependent part, $\vec{F}(x)$, and an operator-dependent part, $\vec{\Lambda}(p)$. This is the kind of equation we are looking for, since it can be applied to different problems provided that the problem-dependent part for each one is computed and we do not need to re-compute the operator-dependent part. In the same way, it can also be applied to any parameter of the operator (value of $p$) without recomputing the problem-dependent part. 

We should notice here that the first column of matrix $\vec{F}(x)$ provides the statistical moments of the fitness distribution in the whole search space considering a uniform random distribution. Thus, $\vec{F}_{1,0}(x)=f_{[0]}(x)$ is the average value of the evaluation function in the search space, $\vec{F}_{2,0}(x)=f^2_{[0]}(x)$ is the second order moment, and so on. We can prove this by setting $p=1/2$ in (\ref{eqn:muc-vector}) because a probability of $p=1/2$ for bit-flip mutation is equivalent to a uniform random selection of a solution in the search space. All the elements but the first in $\vec{\Lambda}(p)$ vanish and we get the claimed result.

\subsection{Computing the Matrix Function $\vec{F}$}
\label{subsec:computing-f}

The computation of the matrix function $\vec{F}$ is not efficient in general. \cite{Sutton2011tcs} provide an algorithm to compute the Walsh decomposition of $f^m$. Using this Walsh decomposition it is possible to obtain the elementary components of $f^m$, as required for the computation of $\vec{F}(x)$. If the Walsh decomposition of $f$ is:
\[
f = \sum_{w \in \Bo^n} a_w \psi_w ,
\]
then the Walsh decomposition of the $m$-th power $f^m$ is:
\[
f^m = \left(\sum_{w \in \Bo^n} a_w \psi_w \right)^m = \sum_{\sum_{w \in \Bo^n}i_w=m} \comb{m}{i_{00\ldots0},i_{00\ldots1},\ldots,i_{11\ldots1}} \prod_{w \in \Bo^n} a_w^{i_w} \psi_{w}^{i_w} .
\]

This procedure has the advantage that it is general and can be used with any function defined over bitstrings. The drawback, however, is its inefficiency when $m$ is high. Thus, for each particular problem, we should analyze the objective function in order to find an efficient way of evaluating the matrix function $\vec{F}$ in an arbitrary solution $x$. In Section~\ref{sec:study} we analyze two problems and provide an efficient computation of this matrix function for these problems.

In some cases, the efficient (polynomial time) evaluation of $\vec{F}(x)$ can only be possible if $\NP=\P$. This happens for example in the SAT problem as the following theorem states.

\begin{proposition}
Let us consider the SAT problem and an evaluation function $f(x)$ that takes value 1 if $x \in \Bo^n$ satisfies the propositional formula and 0 otherwise. If there exists a polynomial time algorithm for computing $f_{[0]}$ then $\NP=\P$.
\end{proposition}
\begin{proof}
The value $f_{[0]}$ is the average value of the objective function in the whole search space. Since $f$ can only take values 0 and 1, if $f_{[0]}>0$ then the formula is satisfiable. Thus, if we find a polynomial time algorithm to evaluate $f_{[0]}$ we can solve the decision problem in polynomial time. But, as SAT is $\NP$-complete then $\NP=\P$.
\end{proof}

As a consequence of the previous proposition we cannot ensure that an efficient evaluation of the matrix function $\vec{F}(x)$ exists in general. The complexity of computing $\vec{F}(x)$  depends on the problem.



\subsection{Fitness Probability Distribution}
\label{subsec:fitness-prob-distr}

With the help of the moments vector $\vec{\mu}\{f(M_p(x))\}$ we can compute the probability distribution of the values of $f$ in a mutated solution. In order to do this we proceed in the same way as \cite{Sutton2011foga}. 

Let us call $\xi_0 < \xi_1 < \cdots < \xi_{q-1}$ to the $q$ possible values that the function $f$ can take in the search space. Since we are dealing with a finite search space, $q$ is a finite number (perhaps very large). We are interested in computing $\Prob{f(M_p(x))=\xi_i}$ for $0 \leq i < q$. In order to simplify the notation in the following we define the vector of probabilities $\vec{\pi}(f(M_p(x)))$ as $\vec{\pi}_i(f(M_p(x)))=\Prob{f(M_p(x))=\xi_i}$.

\begin{theorem}
Let us consider the binary hypercube and let us denote with $\xi_i$ the possible values that the objective function $f$ can take in the search space, where $\xi_i < \xi_{i+1}$ for $0\leq i < q-1$. Then, the vector of probabilities $\vec{\pi}(f(M_p(x)))$ can be computed as:
\begin{equation}
\label{eqn:fitness-probdist}
\vec{\pi}(f(M_p(x))) = \underbrace{\left(\vec{V}^{T}\right)^{-1} \vec{F}(x)}_{\mbox{problem-dependent}} \vec{\Lambda}(p) ,
\end{equation}
where the matrix function $\vec{F}(x)$ is limited to the first $q$ rows and $\vec{V}$ denotes the Vandermonde matrix for the $\xi_i$ values, that is, $\vec{V}_{i,j}=\xi_i^j$ for $0 \leq i,j < q$.
\end{theorem}
\begin{proof}
We can compute the $m$-th moment $\vec{\mu}_m\{f(M_p(x))\}$ using the following expression: 
\begin{equation}
\vec{\mu}_m\{f(M_p(x))\} = \sum_{i=0}^{q-1} \xi_i^m \Prob{f(M_p(x))=\xi_i}.
\end{equation}
We can write this in vector form as:
\begin{align}
\vec{\mu}\{f(M_p(x))\} &= \left( \begin{array}{cccc} 1 & 1 & \cdots & 1 \\
\xi_0^1 & \xi_1^1 & \cdots & \xi_{q-1}^1 \\
\vdots & \vdots & \ddots & \vdots \\
\xi_0^{q-1} & \xi_1^{q-1} & \cdots & \xi_{q-1}^{q-1} \\ \end{array} \right) \left( \begin{array}{c} \Prob{f(M_p(x))=\xi_0} \\ \Prob{f(M_p(x))=\xi_1} \\ \vdots \\ \Prob{f(M_p(x))=\xi_{q-1}} \\ \end{array} \right) \notag \\
&= \vec{V}^{T} \vec{\pi}(f(M_p(x))) .
\end{align}
Using (\ref{eqn:muc-vector}) we can write:
\begin{align}
\vec{\mu}\{f(M_p(x))\} = \vec{V}^{T} \vec{\pi}(f(M_p(x))) = \vec{F}(x) \vec{\Lambda}(p) ,
\end{align}
and solving $\vec{\pi}(f(M_p(x)))$ we finally get (\ref{eqn:fitness-probdist}). The determinant of the Vandermonde matrix is $\prod_{0 \leq i < j < q} \left(\xi_{i}-\xi_{j}\right)$~\citep[pp. 17-18]{Mirsky1955} and the matrix is nonsingular if and only if all the $\xi_i$ values are different. This is our case, so the Vandermonde matrices we use are invertible.
\end{proof}

Again we can observe that $\vec{\pi}(f(M_p(x)))$ is the product of a term that is problem-dependent and a vector that depends on the parameter of the mutation $p$. From (\ref{eqn:fitness-probdist}) it is clear that each particular probability $\vec{\pi}_i(f(M_p(x)))$ is a polynomial in $p$.

We can also compute the cumulative density function $\vec{\Pi}(f(M_p(x)))$ defined by:
\begin{equation}
\vec{\Pi}_i(f(M_p(x))) = \Prob{f(M_p(x)) \leq \xi_i} = \sum_{j=0}^{i} \vec{\pi}_j(f(M_p(x))) .
\end{equation}

We can write the previous equation in vector form as:
\begin{equation}
\label{eqn:cdf-pi}
\vec{\Pi}(f(M_p(x))) = \vec{L} \vec{\pi}(f(M_p(x))) .
\end{equation}
where $\vec{L}$ is the lower triangular matrix defined by 
\[
\vec{L}_{i,j} = \left\{\begin{array}{ll}1 & \mbox{if $i \geq j$,} \\ 0 & \mbox{otherwise.}\end{array}\right.
\]

We can notice again that each element of $\vec{\Pi}(f(M_p(x)))$ is a polynomial in $p$. The component $\vec{\Pi}_i(f(M_p(x)))$ is the probability of reaching a solution $y$ with function value $f(y) \leq \xi_i$ after the mutation with parameter $p$. If $x$ has function value $f(x)=\xi_i$, then $1-\vec{\Pi}_i(f(M_p(x)))$ is the probability of improving the function value of solution $x$ in one application of bit-flip mutation. For problems in which the matrix $\vec{F}$ can be efficiently computed the expression $1-\vec{\Pi}_i(f(M_p(x)))$ could be used as the base for a new mutation operator that tries to maximize the probability of an improving move. 

\section{Case Studies}
\label{sec:study}

In this section we present the elementary landscape decomposition of two well-known problems and their powers (the $\vec{F}(x)$ matrix). With this decomposition we can compute the probability distribution of any solution after mutation. We start by analyzing a toy problem: Onemax. In Section~\ref{subsec:max-sat} we analyze MAX-SAT.

\subsection{Onemax}
\label{subsec:onemax}

Onemax is a linear pseudo-Boolean fitness function that is often used in the analysis of evolutionary algorithms. In our case, we consider the sum of all order-1 Walsh functions, which is related to Onemax by a simple linear transformation. That is:
\begin{equation}
\label{eqn:f-onemax}
f(x) = \sum_{i=1}^{n} \psi_{\underline{i}}(x) =  n - 2 \sum_{i=1}^{n} x_i = n - 2 |x| .
\end{equation}

The objective function in Onemax is $|x|$ (the number of ones in $x \in \Bo^n$).
Maximizing the number of ones in $x$ (original Onemax problem) is equivalent to minimizing $f(x)$.
We should notice here that $f(x)$ can take values in the range $[-n,n]$ by steps of~$2$. That is, the range of $f$ is the set $\{n-2j|j\in\mathbb{N}, 0\leq j \leq n\}$. Although we study here the function $f(x)$ defined in~(\ref{eqn:f-onemax}) for the sake of simplicity, we will see at the end of this section that the probability distribution after mutation of the regular Onemax function is the same as $f(x)$.

The following lemma provides intermediate results that will be useful in the search for an expression for  $\vec{F}(x)$.

\begin{lemma}
\label{lem:sum-walsh}
The sum of all the Walsh functions with the same order is related to the Krawtchouk matrices by means of the following identity:
\begin{equation}
\sum_{w \in \Bo^n \atop |w| = p} \psi_{w}(x) = \krawel{(n)}{p}{|x|}.
\end{equation}
\end{lemma}
\begin{proof}
The claim follows immediately from Eq. (\ref{eqn:sum-walsh-r}) when $t=11\dots 1$.
\end{proof}

\begin{theorem}
\label{thm:fmatrix-onemax}
The matrix function $\vec{F}(x)$ for the objective function $f(x)$ defined in (\ref{eqn:f-onemax}) depends only on $|x|$ and its elements satisfy the following identity:
\begin{equation}
\label{eqn:matrix-f-onemax}
\vec{F}_{m,j}(x) = \vec{\Xi}^{(n)}_{m,j} \krawel{(n)}{j}{|x|}  ,
\end{equation}
where $\kraw{(n)}$ is the $n$-th Krawtchouk matrix and $\vec{\Xi}^{(n)}$ is the matrix defined as:
\begin{equation}
\label{eqn:xi}
\vec{\Xi}^{(n)}_{m,j} = \frac{1}{2^n} \sum_{k=0}^n (n-2k)^m \krawel{(n)}{k}{j} .
\end{equation}
\end{theorem}

\begin{proof}
Let us write the Walsh decomposition of $f^m$. Given a binary string $w \in \Bo^n$, the Walsh coefficient $a^{(m)}_{w}$ of $f^m$ is
\begin{align}
a^{(m)}_{w} &= \frac{1}{2^n}\sum_{x \in \Bo^n} \psi_{w} (x) f^m(x) 
= \frac{1}{2^n}\sum_{x \in \Bo^n} \psi_{w} (x) \left(n-2|x|\right)^m  ~~~\text{dividing the search space} \notag \\
&= \frac{1}{2^n} \sum_{k=0}^{n} \left(n-2k\right)^m \sum_{x \in \Bo^n \atop |x|=k} \psi_{w} (x) = \frac{1}{2^n} \sum_{k=0}^{n} \left(n-2k\right)^m \krawel{(n)}{k}{|w|} = \vec{\Xi}^{(n)}_{m,|w|} ,
\end{align}
where we used the result of Lemma~\ref{lem:sum-walsh} and introduced the matrix $\vec{\Xi}^{(n)}$ to simplify the notation. Now we can sum together all the Walsh functions of the same order $j$ to find the elementary component $f^m_{[j]}$:
\begin{align}
\label{eqn:fmatrix-onemax}
\vec{F}_{m,j} (x) = f^m_{[j]}(x) = \sum_{w \in \Bo^n \atop |w| = j}  a^{(m)}_{w} \psi_{w}(x) = 
\vec{\Xi}^{(n)}_{m,j} \sum_{w \in \Bo^n \atop |w| = j}  \psi_{w}(x) = \vec{\Xi}^{(n)}_{m,j} \krawel{(n)}{j}{|x|}  ,
\end{align}
where we used Lemma~\ref{lem:sum-walsh} in the last step.
\end{proof}

In the following proposition we provide a property of the $\vec{\Xi}^{(n)}$ matrix that is useful to simplify the computation of the matrix.

\begin{proposition}
All the elements $\vec{\Xi}^{(n)}_{m,j}$ in which $m + j$ is odd are zero. 
\end{proposition}
\begin{proof}
We can develop (\ref{eqn:xi}) to write:
\begin{align}
2\vec{\Xi}^{(n)}_{m,j} &= \frac{1}{2^n} \sum_{k=0}^n \left((n-2k)^m \krawel{(n)}{k}{j} + (n-2k)^m \krawel{(n)}{k}{j}\right) & \text{changing $k$ by $n-k$} \nonumber \\
&= \frac{1}{2^n} \sum_{k=0}^n \left((n-2k)^m \krawel{(n)}{k}{j} + (2k-n)^m \krawel{(n)}{n-k}{j}\right) \nonumber \\
&= \frac{1}{2^n} \sum_{k=0}^n (n-2k)^m \left( \krawel{(n)}{k}{j} + (-1)^m \krawel{(n)}{n-k}{j}\right) & \text{by Proposition~\ref{prop:kraw-col}, Eq. (\ref{eqn:kraw-row})} \nonumber \\
&= \frac{1}{2^n} \sum_{k=0}^n (n-2k)^m \left( \krawel{(n)}{k}{j} + (-1)^{m+j} \krawel{(n)}{k}{j}\right) \nonumber \\
&= \frac{1}{2^n} \sum_{k=0}^n (n-2k)^m \krawel{(n)}{k}{j} \left(1  + (-1)^{m+j} \right) . \nonumber
\end{align}
If $m+j$ is odd all the terms in the sum are zero and, thus, $\vec{\Xi}^{(n)}_{m,j}=0$.
\end{proof}

Theorem~\ref{thm:fmatrix-onemax} claims that $\vec{F}$ depends only on $|x|$ and not on the solution itself. As a consequence, the vector of probabilities $\vec{\pi}(f(M_p(x)))$ depends only on $|x|$. But, according to (\ref{eqn:f-onemax}), $|x|$ is related to the fitness value of a solution by $f(x)=n-2|x|$, and the vector of probabilities $\vec{\pi}$ depends only on the fitness level of the solution we are evaluating. We can then build a matrix, denoted with $\vec{\varpi}$, where element $\vec{\varpi}_{i,j}$ is the probability of generating a solution with fitness $\xi_j$ using bit-flip mutation from a solution with fitness $\xi_i$. This matrix depends on $p$ (probability of mutation), but we omit $p$ in the notation to make it simpler. The expression for $\vec{\varpi}_{i,j}$ can be obtained using simple counting arguments, without the need of the mathematical framework developed in Section~\ref{sec:mutation}. However, in the next theorem we provide an expression for this matrix using our mathematical framework. The purpose of this result is twofold: it proves that  $\vec{\varpi}_{i,j}$ can be computed using our framework and, to the best of our knowledge, it provides a
previously unknown expression for $\vec{\varpi}_{i,j}$ involving Krawtchouk matrices. 

\begin{theorem}
Given the objective function defined in (\ref{eqn:f-onemax}) over $\Bo^n$, the probability of reaching a solution with fitness $\xi_j=2j-n$ when bit-flip mutation with probability $p$ is applied to a solution with fitness $\xi_i=2i-n$ is given by:
\begin{equation}
\label{eqn:varpi-onemax}
\vec{\varpi}_{i,j} = \sum_{l=0}^{n} \krawel{(n)}{j}{l} (1-2p)^l \krawel{(n)}{l}{i} ,
\end{equation}
where $0 \leq i,j \leq n$.
\end{theorem}
\begin{proof}
First, we will express the matrix $\vec{\Xi}^{(n)}$ as a product of two other matrices.
\begin{align}
\vec{\Xi}^{(n)}_{c,j} &= \frac{1}{2^n} \sum_{k=0}^n (n-2k)^c \krawel{(n)}{k}{j} ~~~~ \text{replacing $k$ by $n-k$} \nonumber \\
&= \frac{1}{2^n} \sum_{k=0}^n (2k-n)^c \krawel{(n)}{n-k}{j} = \frac{1}{2^n} \sum_{k=0}^n \xi_k^c \krawel{(n)}{n-k}{j} = \frac{1}{2^n} \sum_{k=0}^n \vec{V}_{k,c} \krawelt{(n)}{k}{j} \nonumber \\
&= \frac{1}{2^n} \left(\vec{V}^T \krawt{(n)}\right)_{c,j} ,
\end{align}
where we used the Vandermonde matrix and we introduced a new matrix $\krawt{(n)}$. This is the Krawtchouk matrix of order $n$ in which the rows are reversed. Then, we have $\vec{\Xi}^{(n)} = 2^{-n} \vec{V}^T \krawt{(n)}$. Let us now define the vector $\krawel{(n)}{*}{j}$ as the $j$-th column of the $n$-th order Krawtchouk matrix. We can write the matrix function $\vec{F}(x)$ defined in (\ref{eqn:fmatrix-onemax}) in a compact way as:
\begin{align}
\vec{F}(x) = \vec{\Xi}^{(n)} \Diag{\krawel{(n)}{*}{|x|}} ,
\end{align}
where the function $\Diag{}$ maps a vector into a matrix having the vector in the diagonal.
If we introduce this compact expression of $\vec{F}(x)$ in (\ref{eqn:fitness-probdist}) we obtain:
\begin{align}
\vec{\pi}(f(M_p(x))) &= \left(\vec{V}^T\right)^{-1} \vec{F}(x) \vec{\Lambda}(p) = \left(\vec{V}^T\right)^{-1} \vec{\Xi}^{(n)} \Diag{\krawel{(n)}{*}{|x|}} \vec{\Lambda}(p) \nonumber \\
	&= \frac{1}{2^n} \left(\vec{V}^T\right)^{-1} \vec{V}^T \krawt{(n)} \Diag{\krawel{(n)}{*}{|x|}} \vec{\Lambda}(p) = \frac{1}{2^n} \krawt{(n)} \Diag{\krawel{(n)}{*}{|x|}} \vec{\Lambda}(p) \nonumber \\
\label{eqn:pi-matrix} &= \frac{1}{2^n} \krawt{(n)} (\vec{\Lambda}(p) \circ \krawel{(n)}{*}{|x|}).
\end{align}
where the symbol $\circ$ denotes the Hadamard product\footnote{The Hadamard product of two matrices with the same dimension is the element-wise product of the matrices.} of matrices and we used the fact that $\Diag{\vec{A}} \vec{B} = \vec{A} \circ \vec{B}$. 

According to the definition of $\vec{\varpi}$, it must be related to $\vec{\pi}$ by the following equation:
\begin{equation}
\vec{\varpi}_{i,j} = \vec{\pi}_j (f(M_p(x))) ~~ \text{for any $x$ such that $f(x)=\xi_i$.}
\end{equation}

Since $\xi_i=2i-n$, $f(x)=n-2|x|=\xi_i$ if and only if $|x|=n-i$ and using (\ref{eqn:pi-matrix}) we have:
\begin{align}
\vec{\varpi}_{i,j} &= \vec{\pi}_j (f(M_p(x))) ~~ \text{for any $x$ with $|x|=n-i$}\nonumber  \\
&= \frac{1}{2^n} \left(\krawt{(n)} (\vec{\Lambda}(p) \circ \krawel{(n)}{*}{n-i})\right)_j 
= \frac{1}{2^n} \sum_{l=0}^{n} \krawelt{(n)}{j}{l} (\vec{\Lambda}(p) \circ \krawel{(n)}{*}{n-i})_l \nonumber \\
&= \frac{1}{2^n} \sum_{l=0}^{n} \krawelt{(n)}{j}{l} \vec{\Lambda}_l(p) \krawel{(n)}{l}{n-i} 
= \frac{1}{2^n} \sum_{l=0}^{n} \krawel{(n)}{n-j}{l} \vec{\Lambda}_l(p) \krawel{(n)}{l}{n-i} & \text{by Proposition~\ref{prop:kraw-col}} \nonumber \\
&= \frac{1}{2^n} \sum_{l=0}^{n} \krawel{(n)}{j}{l} (-1)^l\vec{\Lambda}_l(p) \krawel{(n)}{l}{n-i} 
= \frac{1}{2^n} \sum_{l=0}^{n} \krawel{(n)}{j}{l} \vec{\Lambda}_l(p) \krawel{(n)}{l}{i} ,
\end{align}
and we get (\ref{eqn:varpi-onemax}) just considering that $\vec{\Lambda}_l(p) = (1-2p)^l$.
\end{proof}

In the following proposition we provide two properties of the $\vec{\varpi}$ matrix that are useful to reduce the computational complexity of its computation.

\begin{proposition}
\label{prop:varpi}
The matrix $\vec{\varpi}$ has the following properties:
\begin{eqnarray}
\label{eqn:quasi-sim-varpi}
&& \comb{n}{i} \vec{\varpi}_{i,j} = \comb{n}{j} \vec{\varpi}_{j,i} , \\
&& \vec{\varpi}_{n-i,n-j}=\vec{\varpi}_{i,j} ,
\end{eqnarray}
where $0 \leq i,j \leq n$.
\end{proposition}
\begin{proof}
The first property is a consequence of an analogous property of the Krawtchouk matrices: $\comb{n}{j} \krawel{(n)}{i}{j}=\comb{n}{i} \krawel{(n)}{j}{i}$ \cite[p. 179]{Terras1999}. We can write:
\begin{align}
\comb{n}{i} \vec{\varpi}_{i,j} &= \frac{1}{2^n} \sum_{l=0} \krawel{(n)}{j}{l} \vec{\Lambda}_l(p) \comb{n}{i} \krawel{(n)}{l}{i} 
= \frac{1}{2^n} \sum_{l=0} \comb{n}{l} \krawel{(n)}{j}{l} \vec{\Lambda}_l(p) \krawel{(n)}{i}{l} \nonumber \\
&= \frac{1}{2^n} \sum_{l=0} \comb{n}{j} \krawel{(n)}{l}{j} \vec{\Lambda}_l(p) \krawel{(n)}{i}{l} 
= \comb{n}{j} \vec{\varpi}_{j,i}. \nonumber
\end{align}

The second property is a consequence of Proposition~\ref{prop:kraw-col}:
\begin{align}
\vec{\varpi}_{n-i,n-j} &= \frac{1}{2^n} \sum_{l=0} \krawel{(n)}{n-j}{l} \vec{\Lambda}_l(p) \krawel{(n)}{l}{n-i}
= \frac{1}{2^n} \sum_{l=0} (-1)^l \krawel{(n)}{j}{l} \vec{\Lambda}_l(p) (-1)^l \krawel{(n)}{l}{i} \nonumber \\
&= \frac{1}{2^n} \sum_{l=0} \krawel{(n)}{j}{l} \vec{\Lambda}_l(p) \krawel{(n)}{l}{i} = \vec{\varpi}_{i,j}.\nonumber
\end{align}
\end{proof}

At this point we can discuss the utility of the $\vec{\varpi}$ matrix. We can see $\vec{\varpi}$ as a practical substitute for all the probability vectors $\vec{\pi}(M_p(f(x)))$ in the case of the objective function $f(x)$ defined in (\ref{eqn:f-onemax}). In general, the components of the previous vector depend on the solution $x$. However, in the particular case of the Onemax-related function (\ref{eqn:f-onemax}) the components depend only on the fitness level $\xi_i$ the solution has. This way we can forget the concrete solution $x$ and focus only on the fitness levels $\xi_i$. Furthermore, the number of fitness levels is $n+1$ and the complexity of computing any element of $\vec{\varpi}$ using (\ref{eqn:varpi-onemax}) is $O(n)$, where we assume that the Krawtchouk matrix $\kraw{(n)}$ is precomputed\footnote{Krawtchouk matrix $\kraw{(n)}$ can be precomputed in $O(n^3)$ using Proposition 2.1 of~\cite{Feinsilver2005}.}. This means that we can compute the probabilities of reaching any fitness level from any other one after bit-flip mutation in $O(n^3)$. That is, we obtain in polynomial time a practical piece of information that summarizes the behavior of bit-flip mutation in this problem. We will see in Section~\ref{sec:applications} how this information can be used.

We derived the $\vec{\varpi}$ matrix for only one objective function. Now we wonder if similar $\vec{\varpi}$ matrices can be derived for other objectives functions. The answer to this question is not easy in general, but the next results gives a first answer in this line. Let us first formally define the property that allows one to compute a matrix like $\vec{\varpi}$.

\begin{definition}
Let $f(x)$ be an objective function and let us call $\xi_0 < \xi_1 < \ldots < \xi_{q-1}$ to the different values it can take. We say that the function $f$ \emph{has a fitness-dependent distribution} for a unary operator if the probability distribution of the objective value after applying the operator to any solution does only depend on the objective value of the initial solution. In formal terms, if $U(x) \in \Bo^n$ is a random variable that represents the application of the unary $U$ operator to $x$, we have
\begin{equation}
\forall x,y \in \Bo^n, f(x)=f(y) \Rightarrow \Prob{f(U(x))=\xi_j} = \Prob{f(U(y))=\xi_j} ,
\end{equation}
for all the possible $\xi_j$ values. If this happens, then we can define a matrix $\vec{\varpi}$ whose elements are:
\begin{equation}
\vec{\varpi}_{i,j} = \ProbCond{f(U(x))=\xi_j}{f(x)=\xi_i} .
\end{equation}
\end{definition}

There is a trivial family of functions having fitness-dependent distributions for any unary operator. It is the family of injective functions. In these functions each particular solution has a unique image and the fitness-dependency condition trivially holds. However, the probability matrix in this case has size $2^n \times 2^n$ (the size of the search space squared), what makes this treatment impractical. Even simple linear pseudo-Boolean functions (such as BINVAL) can have this property. 
The next theorem claims that the property of having a fitness-dependent distribution can be kept even after some simple manipulations of the fitness function.

\begin{theorem}
\label{thm:transformation}
Let $g(x)$ be an objective function having a fitness-dependent distribution for the unary operator $U$ and let us call $\vec{\varpi}^{(g)}$ the associated probability matrix, where we used the name of the function as superindex.
Then, the function $f(x)$, which is a composition of $g(x)$ with another function, also has a fitness-dependent distribution for $U$ under the following conditions:
\begin{itemize}
\item When $f(x)=h(g(x))$ for $h$ a strictly increasing function. The probability matrix does not change: $\vec{\varpi}^{(f)}=\vec{\varpi}^{(g)}$.
\item When $f(x)=h(g(x))$ for $h$ a strictly decreasing function. The probability matrix flips its rows and columns: $\vec{\varpi}^{(f)}_{i,j}=\vec{\varpi}^{(g)}_{(q-1)-i,(q-1)-j}$.
\item When $f(x)=g(x \oplus u)$ for $u \in \Bo^n$ and $U$ commute with the $\oplus$ operator: $\Prob{U(x\oplus u)=y}=\Prob{U(x)\oplus u= y} \forall x,u \in \Bo^n$. The probability matrix does not change: $\vec{\varpi}^{(f)}=\vec{\varpi}^{(g)}$.
\end{itemize}
\end{theorem}
\begin{proof}
First, we can observe that in the three cases the number of values that $f$ can take is the same as the number of values that $g$ can take, $|f(\Bo^n)|=|g(\Bo^n)|=q$. Then, let us denote with $\xi_0^{(g)} < \xi_1^{(g)} < \ldots < \xi_{q-1}^{(g)}$ these values for the $g$ function. We will use the notation $\xi_i^{(f)}$ to refer to the corresponding values of $f$.

Let us start with the first case: $f(x) = h(g(x))$ and $h$ strictly increasing. In this case $\xi_i^{(f)}=h(\xi_i^{(g)})$ for all $0 \leq i < q$. And the property of having a fitness-dependent distribution trivially holds for $f$ since if $f(x)=f(y)$ then $g(x)=g(y)$ and consequently $\Prob{g(U(x))=\xi_j^{(g)}} = \Prob{g(U(y))=\xi_j^{(g)}}$, what implies $\Prob{f(U(x))=\xi_j^{(f)}} = \Prob{f(U(y))=\xi_j^{(f)}}$. Regarding the probability matrix we have:
\begin{align}
\vec{\varpi}^{(f)}_{i,j} &= \ProbCond{f(U(x))=\xi_j^{(f)}}{f(x)=\xi_i^{(f)}} & \text{by definition of $f$} \nonumber \\
&= \ProbCond{h(g(U(x)))=h(\xi_j^{(g)})}{h(g(x))=h(\xi_i^{(g)})} & \text{since $h$ is strictly monotone} \nonumber \\
&= \ProbCond{g(U(x))=\xi_j^{(g)}}{g(x)=\xi_i^{(g)}} = \vec{\varpi}^{(g)}_{i,j}.
\end{align}

In the second case, in which $f(x) = h(g(x))$ for $h$ a strictly decreasing function, we can prove that $f$ has a fitness-dependent distribution with an argument similar to the first case. However, the probability matrix is different due to the change in the order of the values $\xi_i^{(f)}$. Since $h$ is strictly decreasing we have $\xi_i^{(f)}=h(\xi_{(q-1)-i}^{(g)})$. Thus, the elements of the probability matrix are given by:
\begin{align}
\vec{\varpi}^{(f)}_{i,j} &= \ProbCond{f(U(x))=\xi_j^{(f)}}{f(x)=\xi_i^{(f)}} & \text{by definition of $f$} \nonumber \\
&= \ProbCond{h(g(U(x)))=h( \xi_{(q-1)-j}^{(g)})}{h( g(x))=h(\xi_{(q-1)-i}^{(g)})} & \text{$h$ strictly mon.} \nonumber \\
&= \ProbCond{g(U(x))=\xi_{(q-1)-j}^{(g)}}{g(x)=\xi_{(q-1)-i}^{(g)}} = \vec{\varpi}^{(g)}_{(q-1)-i,(q-1)-j}.
\end{align}

Finally, let us prove the last case. The values that the function takes do not change, that is: $\xi_i^{(f)} = \xi_i^{(g)}$. If $f(x)=f(y)$ then $g(x \oplus u)=g(y \oplus u)$, what implies $\Prob{g(U(x\oplus u))=\xi_j^{(g)}} = \Prob{g(U(y \oplus u))=\xi_j^{(g)}}$ by hypothesis. But if $U$ commutes with $\oplus$ then we have:
\begin{align}
\Prob{g(U(x\oplus u))=\xi_j^{(g)}} &= \Prob{g(U(x)\oplus u)=\xi_j^{(g)}} = \Prob{f(U(x))=\xi_j^{(f)}} ,
\end{align}
where we used the definition $f(x)=g(x \oplus u)$ in the last step and the fact that $\xi_i^{(f)} = \xi_i^{(g)}$.

As a consequence we have $\Prob{f(U(x))=\xi_j^{(f)}} = \Prob{f(U(y))=\xi_j^{(f)}}$ and $f$ has a fitness-dependent distribution for $U$. The elements of the probability matrix are:
\begin{align}
\vec{\varpi}^{(f)}_{i,j} &= \ProbCond{f(U(x))=\xi_j^{(f)}}{f(x)=\xi_i^{(f)}} & \text{by definition of $f$} \nonumber \\
&= \ProbCond{g(U(x)\oplus u)=\xi_{j}^{(g)}}{g(x\oplus u)=\xi_{i}^{(g)}} & \text{by commutation of $U$ and $\oplus$} \nonumber \\
&= \ProbCond{g(U(x \oplus u))=\xi_{j}^{(g)}}{g(x \oplus u)=\xi_{i}^{(g)}} = \vec{\varpi}^{(g)}_{i,j}.
\end{align}

\end{proof}

The only condition imposed to the unary operator in the previous theorem is the commutation with $\oplus$. Fortunately, the bit-flip mutation operator commutes with $\oplus$. Furthermore, we provide in the next proposition a result that generalizes that of the mutation operator.

\begin{proposition}
\label{prop:commutation}
If a unary operator $U$ has the property $\Prob{U(x)=y}=f(x\oplus y)$ for a real function $f$ then it commutes with the $\oplus$ operation.
\end{proposition}
\begin{proof}
For any $x,u,y \in \Bo^n$ we can write:
\begin{align}
\Prob{U(x\oplus u)=y} &= f(x \oplus u \oplus y) = \Prob{U(x) = u \oplus y} =\Prob{U(x)\oplus u= y} ,
\end{align}
and we have the commutation property.
\end{proof}

The bit-flip  mutation satisfies the hypothesis of the previous proposition, as Lemma~\ref{lem:prob-mut} states.
%
%
Now, we can combine the results of Theorem~\ref{thm:transformation} and Proposition~\ref{prop:commutation} to provide a concrete result for the Onemax-related functions.

\begin{proposition}
\label{prop:onemax-family}
All the objective functions of the form
\begin{equation}
\label{eqn:onemax-family}
g(x) = h(|x \oplus u|) ,
\end{equation}
where $h$ is a strictly monotone function have a fitness-dependent distribution for the bit-flip mutation operator and the probability matrix is the one defined in (\ref{eqn:varpi-onemax}).
\end{proposition}
\begin{proof}
First, we observe that in the case of the sum of order-1 Walsh functions (\ref{eqn:f-onemax}) the probability matrix $\vec{\varpi}^{(f)}$ does not change even in the case in which we compose the functions with a strictly decreasing function, since $q=n+1$ and $\vec{\varpi}_{n-i,n-j}=\vec{\varpi}_{i,j}$ according to Proposition~\ref{prop:varpi}. Then, based on the results of Theorem~\ref{thm:transformation} and Proposition~\ref{prop:commutation} we only need to express $g(x)$ as a strictly monotone function of the objective function $f$ defined in (\ref{eqn:f-onemax}). This expression is:
\begin{align}
g(x) = h(|x \oplus u|) = h\left(\frac{n-f(x \oplus u)}{2}\right) .
\end{align}
\end{proof}

A direct consequence of the previous result is that even although we focused in this section on the objective function defined in (\ref{eqn:f-onemax}) instead of the Onemax objective function, the probability matrix $\vec{\varpi}$ is valid also for the original Onemax function. Furthermore, it is also valid for any strictly monotone function composed with the Onemax function.

%

\subsection{MAX-SAT}
\label{subsec:max-sat}


The MAX-SAT problem is a well-known $\NP$-hard problem related to the satisfiability of Boolean formulas. An instance of this problem is composed of a set of clauses $C$. A clause is a disjunction of literals, each one being a decision variable $x_i$ or a negated decision variable $\overline{x_i}$. The MAX-SAT problem consists in finding an assignment of Boolean values to the literals in such a way that the number of satisfied clauses is maximum. Let us assume that there exist $n$ Boolean decision variables.
For each clause $c \in C$ we define the vectors $v(c) \in \Bo^n$ and $u(c) \in \Bo^n$ as follows \citep{Sutton2009}:
\begin{eqnarray}
&& v_i(c) = \left\{\begin{array}{ll}
1 & \mbox{if $x_i$ appears (negated or not) in $c$} ,\\
0 & \mbox{otherwise},
\end{array} \right. \\
&& u_i(c) = \left\{\begin{array}{ll}
1 & \mbox{if $x_i$ appears negated in $c$}, \\
0 & \mbox{otherwise}.
\end{array} \right.
\end{eqnarray}

We will omit the argument of the vectors (the clause) when there is no confusion. According to this definition $u \wedge v = u$. We should note here that the previous notation allows us to express the empty clause, $\square$, with $v=u=0$. But it is not possible to express the top clause $\top$. We will need a special treatment of the top clause in the following.

The objective function of MAX-SAT is defined as
\begin{align}
\label{eqn:f-maxsat}f(x) &= \sum_{c \in C} f_{c}(x); ~~~ \mbox{where} \notag \\
f_c(x) &= \left\{\begin{array}{ll}1 & \mbox{if $c$ is satisfied with assignment $x$}, \\ 
0 & \mbox{otherwise}.\end{array} \right.
\end{align}

A clause $c$ is satisfied with $x$ if at least one of the literals is true (we assume the usual identity true=1 and false=0). Using the vectors $v(c)$ and $u(c)$ we can say that $c$ is satisfied by $x$ if $(\overline{x} \wedge  u)  \vee  (x \wedge v \wedge \overline{u}) \neq 0$.

\cite{Sutton2009} provide the Walsh decomposition for the MAX-SAT problem. Let the function $f_{c}$ evaluate one clause $c \in C$. 
The Walsh coefficients for $f_c$ are:
\begin{equation}
\label{eqn:awm}
a_w = \left\{
\begin{array}{ll}
0 & \mbox{if $w \wedge \bar{v}\neq 0$}, \\
1 - \frac{1}{2^{|v|}} & \mbox{if $w=0$}, \\
\frac{-1}{2^{|v|}}\psi_{w} (u) & \mbox{otherwise}.
\end{array} 
\right.
\end{equation}
If the clause $c$ is $\top$ then the only nonzero Walsh coefficient is $a_0=1$.

For the sake of simplicity in the mathematical development, instead of using $f_c$ in the following, it is better to use $g_c(x)=1-f_c(x)$. The Walsh coefficients for $g_c$ are:
\begin{equation}
\label{eqn:awm-g}
a_w = \left\{
\begin{array}{ll}
0 & \mbox{if $c=\top$ or $w \wedge \bar{v}\neq 0$}, \\
\frac{1}{2^{|v|}}\psi_{w} (u) & \mbox{otherwise}.
\end{array} 
\right.
\end{equation}

We will also focus on the fitness function $g(x)$ defined as:
\begin{equation}
\label{eqn:g-maxsat}
g(x) = \sum_{c \in C} g_c(x) = m - f(x).
\end{equation}
Maximizing $f(x)$ is equivalent to minimizing $g(x)$.


The following lemma provides the elementary landscape decomposition of $g_{c}$. 
\begin{lemma}
\label{lem:g-eld}
The $j$-th elementary component of $g_{c}$ is
\begin{align}
\label{eqn:g-eld}
g_{c,[j]}(x) = (1-\delta_c^\top)  \frac{1}{2^{|v(c)|}} \krawel{(|v(c)|)}{j}{|(x \oplus u(c)) \wedge v(c)|} .
\end{align}
\end{lemma}
\begin{proof}
With the help of (\ref{eqn:awm-g}) we can write
\begin{align}
g_{c,[j]}(x) &= \sum_{w \in \Bo^n \atop |w|=j} a_w \psi_w(x) = 
(1-\delta_c^\top)  \sum_{w \in \Bo^n \wedge v(c) \atop |w|=j} \frac{1}{2^{|v(c)|}}\psi_{w} (u(c)) \psi_w(x) & \text{by (\ref{eqn:walsh-oplus})} \nonumber \\
&= (1-\delta_c^\top)  \sum_{w \in \Bo^n \wedge v(c) \atop |w|=j} \frac{1}{2^{|v(c)|}}\psi_{w} (x\oplus u(c)) & \text{by (\ref{eqn:sum-walsh-r})} \nonumber \\
&= (1-\delta_c^\top)  \frac{1}{2^{|v(c)|}} \krawel{(|v(c)|)}{j}{|(x \oplus u(c)) \wedge v(c)|} .
\end{align}
\end{proof}
%
%

The next two lemmas provide intermediate results related to the $g_c$ functions that are required in the proof of the main theorem in this section.

\begin{lemma}
\label{lem:power-g}
The $m$-th power of $g_{c}$ for $m > 0$ is:
\begin{equation}
g_{c}^m (x) = g_{c}(x) .
\end{equation}
\end{lemma}
\begin{proof}
The function $g_{c}$ takes only values 0 and 1. If $m>0$ we have $g^m_{c}(x) = g_{c}(x)$.
\end{proof}



\begin{lemma}
\label{lem:prod-g}
Given a family of clauses $c \in C$, the product of functions $g_{c}$ is:
\begin{equation}
\prod_{c \in C} g_{c} (x) = g_{\bigvee C}(x) ,
\end{equation}
where $\bigvee C$ is the disjunction of the family of clauses. For the previous expression to be true even in the case in which the family of clauses is empty we define $g_{\square}(x)=1$.
\end{lemma}
\begin{proof}
The function $g_c$ is 0 when the clause $c$ is satisfied. Thus, a product of $g_c$ functions will be 0 when any of the clauses is satisfied and 1 if none of the clauses is. This behavior is the same as the function $g$ associated with the disjunction of the clauses. This disjunction is also another clause. 
\end{proof}

The following theorem provides the expression for the matrix function $\vec{F}(x)$ for $g(x)$ defined in (\ref{eqn:g-maxsat}). 

\begin{theorem}

The matrix function $\vec{F}(x)$ for the objective function $g(x)$ defined in (\ref{eqn:g-maxsat}) is:
\begin{equation}
\label{eqn:matrix-f-maxsat}
\vec{F}_{m,j}(x) = \sum_{W \subseteq C \atop \bigvee W \neq \top} \frac{1}{2^{|v(\bigvee W)|}} \vec{\Upsilon}_{m,|W|} \krawel{(|v(\bigvee W)|)}{j}{|(x \oplus u(\bigvee W))\wedge v (\bigvee W)|},
\end{equation}
where the $\vec{\Upsilon}$ matrix is defined by the following recurrence equations:
\begin{eqnarray}
\label{eqn:upsilon-base}
&& \vec{\Upsilon}_{m,0}=0, \\
\label{eqn:upsilon-rec}
&& \vec{\Upsilon}_{m,k} = k^m - \sum_{l=0}^{k-1} \comb{k}{l} \vec{\Upsilon}_{m,l}.
\end{eqnarray}
\end{theorem}
\begin{proof}
Let us number the clauses in $C$ from $1$ to $h$ and let us denote with $c_j$ the $j$-th clause. We can write $g^m(x)$ as
\begin{align}
g^m(x) &= \left( \sum_{j=1}^{h} g_{c_j}(x)\right)^m = \sum_{i_1+i_2+\ldots+i_h=m} \comb{m}{i_1, i_2, \ldots, i_h} \prod_{j=1}^{h} g_{c_j}^{i_j}(x) \nonumber \\
\intertext{by Lemma~\ref{lem:power-g}}
&= \sum_{i_1+i_2+\ldots+i_h=m} \comb{m}{i_1, i_2, \ldots, i_h} \prod_{j=1 \atop i_j>0}^{h} g_{c_j}(x) \nonumber \\
\intertext{removing the indices that are zero we can write the sum in an alternative way}
&= \sum_{W \subseteq C} \sum_{i_1+i_2+\ldots+i_{|W|}=m \atop i_1,i_2,\ldots i_{|W|} > 0} \comb{m}{i_1,i_2,\ldots,i_{|W|}} \prod_{c \in W} g_{c}(x) \nonumber \\
&=  \sum_{W \subseteq C} \vec{\Upsilon}_{m,|W|} \prod_{c \in W} g_{c}(x) \nonumber \\
\intertext{by Lemma~\ref{lem:prod-g}}
\label{eqn:f-matrix-g-inter}
&=  \sum_{W \subseteq C} \vec{\Upsilon}_{m,|W|} g_{\bigvee W}(x),
\end{align}
where we defined $\vec{\Upsilon}_{m,k}$ as 
\begin{equation}
\label{eqn:upsilon-def}
\vec{\Upsilon}_{m,k} = \sum_{i_1+i_2+\ldots+i_{k}=m \atop i_1,i_2,\ldots i_{k} > 0} \comb{m}{i_1,i_2,\ldots,i_{k}}.
\end{equation}

Using (\ref{eqn:f-matrix-g-inter}) and (\ref{eqn:g-eld}) we obtain (\ref{eqn:matrix-f-maxsat}).

In order to complete the proof we only need to justify the equations (\ref{eqn:upsilon-base}) and (\ref{eqn:upsilon-rec}) based on the definition (\ref{eqn:upsilon-def}). 
In the following we will use the notation $[k]$ to denote the set of numbers from $1$ to $k$.
When $m,k>0$ the sum of the multinomial coefficients is:
\begin{align}
k^m &= \sum_{i_1+i_2+\ldots+i_k=m} \comb{m}{i_1,i_2,\ldots,i_k} \nonumber \\
\intertext{which we can reorganize in the following way}
&= \sum_{J \subseteq [k] \atop J \neq \emptyset} \sum_{i_1+i_2+\ldots+i_{|J|}=m \atop i_1,i_2,\ldots,i_{|J|}>0} \comb{m}{i_1,i_2,\ldots,i_{|J|}} \nonumber \\
&= \sum_{l=1}^{k} \sum_{J \subseteq [k]\atop |J|=l} \sum_{i_1+i_2+\ldots+i_{|J|}=m \atop i_1,i_2,\ldots,i_{|J|}>0} \comb{m}{i_1,i_2,\ldots,i_{|J|}} \nonumber \\
&= \sum_{l=1}^{k} \sum_{J \subseteq [k]\atop |J|=l} \sum_{i_1+i_2+\ldots+i_l=m \atop i_1,i_2,\ldots,i_l>0} \comb{m}{i_1,i_2,\ldots,i_l}
= \sum_{l=1}^{k} \sum_{J \subseteq [k]\atop |J|=l} \vec{\Upsilon}_{m,l} \nonumber \\
&= \sum_{l=1}^{k} \comb{k}{l} \vec{\Upsilon}_{m,l}  .
\end{align}
In order to extend the previous sum to the value $l=0$ we can just define $\vec{\Upsilon}_{m,0}=0$ as in (\ref{eqn:upsilon-base}) and we obtain the recurrence equation in (\ref{eqn:upsilon-rec}). Now we can observe that (\ref{eqn:upsilon-rec}) is valid even in the case in which $k=0$.
\end{proof}

By the definition of $\vec{\Upsilon}_{m,k}$ in (\ref{eqn:upsilon-def}) it is clear that $\vec{\Upsilon}_{m,k}=0$ if $k>m$. As a consequence, the sum in (\ref{eqn:matrix-f-maxsat}) must consider only the subsets $W$ of at most $m$ elements if we are interested in the elementary landscape decomposition of the $m$-th power of $g(x)$. The computation of the $v$ and $u$ vectors in (\ref{eqn:matrix-f-maxsat}) can be done in $O(n m)$ at most, since we have to explore up to $n$ bits of up to $m$ clauses (it can be much less in practice if the number of literals per clause is low). Thus, the complexity of computing the elementary components of $g^m(x)$ is $O(n m |C|^m)$, where we assume that the matrices $\vec{\Upsilon}$ and $\kraw{(n)}$ are precomputed.

%
%
%
%
%
%
%

\section{Connection to Runtime Analysis}
\label{sec:applications}


The results we present in this section represent, to the best of our knowledge, the first connection between landscape theory and runtime analysis, and this is, in fact, the reason why we think they are relevant. They are an application of the results of Sections~\ref{sec:mutation} and~\ref{sec:study} to the computation of the first hitting time of a $(1+\lambda)$ EA. The results themselves are not new or significant for the runtime community. We use the Markov chain framework by \cite{He:Yao2003}, which has important limitations when the goal is to find asymptotic bounds for runtime. With this framework we are able to compute exact expressions for the expected runtime as a function of $p$, the probability of flipping a bit in the mutation, but we are not able to find asymptotic expressions or make conclusions about the runtime when $n$ is large, which is the main goal of the runtime analysis community.

When one is interested in computing bounds for the runtime required by an evolutionary algorithm to solve an optimization problem, it is quite common to analyze the probability of improving a solution in one iteration of the algorithm. This is the way, for example, in which an upper bound of $O(n \log n)$ is derived for the Onemax problem solved with a $(1+1)$ EA using bit-flip with probability $p=1/n$ \citep[p. 39]{Neumann2010}. These probabilities of improvement are not usually exactly computed, but an asymptotic lower bound of the probability is used instead. Thus, an upper bound of the expected runtime is derived instead of a precise expression.

In Section~\ref{sec:mutation} we showed how we can compute the probability distribution of the objective values after mutation. In Section~\ref{sec:study} we found that for a family of functions that includes Onemax, the probability distribution only depends on the value of the fitness function in the current solution and, thus, we could define a matrix $\vec{\varpi}$, that summarizes the behavior of the algorithm in one step. The question we want to answer in this section is, can we use the $\vec{\varpi}$ matrix to provide an expression of the expected runtime of an evolutionary algorithm? In the next subsections we will show how a precise expression for the expected runtime of a $(1+\lambda)$ EA can be derived using the $\vec{\varpi}$ matrix. In Algorithm~\ref{alg:1+lambda-EA} we show the pseudocode of a $(1+\lambda)$ EA, where (abusing notation slightly) we use $M_p(x)$ to denote a random solution that is the result of applying the bit-flip mutation operator to $x$.

\begin{algorithm}[!ht]
\begin{algorithmic}
\STATE $x$ $\leftarrow$ RandomSolution();
\WHILE{$x$ is not a global optimum}
\FOR{$i=1$ to $\lambda$}
\STATE $y \leftarrow M_p(x)$;
\IF{$f(y) > f(x)$}
\STATE $x \leftarrow y$;
\ENDIF 
\ENDFOR
\ENDWHILE
\end{algorithmic}
\caption{Pseudocode of a $(1+\lambda)$ EA.}
\label{alg:1+lambda-EA}
\end{algorithm}

\subsection{Runtime analysis of $(1+\lambda)$-EAs}

Based on the probability matrix $\vec{\varpi}$, that only assumes one single application of the mutation operator, we can define a new probability matrix $\vec{\varpi}^{(\lambda)}$ related to the generation of $\lambda$ offspring using bit-flip mutation and selecting the best one. This new probability matrix must reduce to $\vec{\varpi}$ when $\lambda=1$. The element $\vec{\varpi}^{(\lambda)}_{i,j}$ is the probability of obtaining a solution with fitness value $\xi_j$ after applying bit-flip mutation $\lambda$ times with probability $p$ to a solution with fitness value $\xi_i$ and taking the offspring with the highest fitness. The following proposition provides an expression for $\vec{\varpi}^{(\lambda)}_{i,j}$.

\begin{proposition}
The probability matrix $\vec{\varpi}^{(\lambda)}$ is defined as
\begin{equation}
\label{eqn:varpi-lambda}
\vec{\varpi}^{(\lambda)}_{i,j} = \left(\sum_{l=0}^j \vec{\varpi}_{i,l} \right)^\lambda  - \left(\sum_{l=0}^{j-1} \vec{\varpi}_{i,l}\right)^\lambda .
\end{equation}
\end{proposition}
\begin{proof}
The element $\vec{\varpi}^{(\lambda)}_{i,j}$ is exactly the probability of obtaining at least one solution with fitness value $\xi_j$ and no solution with a higher fitness value in the $\lambda$ trials. This is exactly the probability of obtaining the $\lambda$ solutions with fitness value lower than or equal to $\xi_j$ (first term) minus the probability of obtaining the solutions with fitness value lower than $\xi_{j}$ (second term).
\end{proof}

We can observe in (\ref{eqn:varpi-lambda}) that $\vec{\varpi}^{(1)}=\vec{\varpi}$. We must recall here that $\vec{\varpi}^{(\lambda)}$ is a polynomial in $p$ because $\vec{\varpi}$ is also a polynomial in $p$. Now we analyze the runtime of the $(1+\lambda)$ EA with the help of the Markov chain framework presented by \cite{He:Yao2003}. The first step is to present the transition matrix (we assume maximization):
\begin{equation}
\label{eqn:transition-matrix}
\vec{P}^{(\lambda)}_{i,j} = \left\{\begin{array}{cc}
\vec{\varpi}^{(\lambda)}_{i,j} & \text{if $j > i$,} \\
\sum_{l=0}^{j} \vec{\varpi}^{(\lambda)}_{i,l} & \text{if $i = j$,}\\
0 & \text{if $j < i$,}
\end{array}\right.
\end{equation}
where $0 \leq i,j < q$. If the probability $p$ of flipping a bit is $0 < p < 1$, then the previous transition matrix will have only one absorbing state that corresponds to the solutions with the highest fitness value $\xi_{q-1}$.

Now we can use some results from the Markov chain theory \citep{Iosifescu1980} to compute the expected runtime of the $(1+\lambda)$ EA. The $\vec{P}$ matrix can be written in the form:
\begin{equation}
\left(
\begin{array}{cc}
\vec{T} & \vec{R} \\
\vec{0} & 1 \\
\end{array}
\right) ,
\end{equation}
where $\vec{R}$ is a column vector and $\vec{T}$ is a $(q-1)\times (q-1)$ submatrix with the transition probabilities of the transient states in the Markov chain. The fundamental matrix is $\vec{N}=(\vec{I}-\vec{T})^{-1}$, and the expected runtime (number of iterations) of the $(1+\lambda)$ EA starting in a solution with fitness value $\xi_i$ and $0 \leq i < q-1$ is given by the $i$-th component of the vector of mean absorption times $\vec{t}_i$. This vector is computed as $\vec{t} = \vec{N} \vec{1}$. From a computational point of view, the vector can be efficiently computed by solving the following linear equation system:
\begin{equation}
(\vec{I} - \vec{T}) \vec{t} = \vec{1} ,
\end{equation}
since $\vec{I}-\vec{T}$ is an upper triangular matrix and the system can be solved in $O(q^2)$. The components of $\vec{t}$ will be, in general, fractions of polynomials in $p$. The vector $\vec{t}$ has only $q-1$ components indexed by number from 0 to $q-2$, but for the sake of completeness we can extend it with an additional component $\vec{t}_{q-1}=0$, which is the expected runtime of the algorithm when the initial solution is the global optimum.

Assuming that the algorithm starts from a random solution, the expected runtime is given by
\begin{equation}
\label{eqn:expected-1+l}
\Exp{\tau} = \sum_{l=0}^{q-1} \frac{|X_l|}{|X|} \vec{t}_l ,
\end{equation}
where $X$ is the set of solutions and $X_l$ is the set of solutions with fitness $f(x)=\xi_l$.

The expected runtime (\ref{eqn:expected-1+l}) is not an approximation or bound, it is the exact expression of the expected runtime as a function of $p$, the probability of flipping a bit. However, this expression will only be practical if 1) we can define the $\vec{\varpi}$ matrix in the problem we are interested and 2) the evaluations of this matrix can be efficiently done in a computer. These conditions limit the number of problems whose runtime can be analyzed using this approach. However, we found in Section~\ref{subsec:onemax} that for any monotone function of Onemax we can efficiently construct and evaluate the $\vec{\varpi}$ matrix. In the next subsection we focus on Onemax.

\subsection{Runtime of $(1+\lambda)$ EA for Onemax}

The Onemax problem has been studied in the literature on runtime analysis many times. \cite{Garnier1999} derived an expression for the transition probability matrix of the $(1+1)$ EA for the Onemax function that was later reported by \cite{He:Yao2003}. Their expression is the same as (\ref{eqn:transition-matrix}) for $\lambda=1$. Tight upper and lower bounds have been derived for the $(1+1)$ EA using different mutation rates. 
Recently, \citet{Witt2013tight} proved that the $(1+1)$ EA optimizes all linear functions (including Onemax) in expected time $en \ln n + O(n)$, and the expected optimization time is polynomial as long as $p = O((\log n)/n)$ and $p = \Omega(1/poly(n))$. \citet{Jansen2005} proved that using a $(1+\lambda)$ EA the expected number of iterations after reaching the global optimum is $O(n \log n/\lambda +n)$.
In summary, the Onemax problem is well-known and the content of this section adds not too much to the current knowledge on this problem. The goal of this section is, thus, to obtain the same results from a different perspective, that of landscape analysis. The advantage of this approach is that with an exact expression of the expected runtime we can find a precise answer to some concrete questions for some particular instances. The disadvantage is that the expression is quite complex to analyze and we need to use numerical methods, so it is not easy to generalize the answers obtained.

Let us first start by studying the $(1+1)$ EA. Taking into account the $\vec{\varpi}$ matrix defined in (\ref{eqn:varpi-onemax}) for Onemax, the expected number of iterations can be exactly computed as a function of $p$, the probability of flipping a bit. Just for illustration purposes, we present the expressions of such expectation for $n\leq 3$:
\begin{align}
\label{eqn:exp-n1}\Exp{\tau} &= \frac{1}{2 p} & \text{for $n=1$,} \\
\label{eqn:exp-n2}\Exp{\tau} &= \frac{7-5 p}{4 (p-2) (p-1) p} & \text{for $n=2$,} \\
\label{eqn:exp-n3}\Exp{\tau} &= \frac{26 p^4-115 p^3+202 p^2-163 p+56}{8 (p-1)^2 p \left(p^2-3 p+3\right) \left(2 p^2-3 p+2\right)} & \text{for $n=3$.}
\end{align}

We can observe how the expressions grow very fast as $n$ increases. The factor $p(p-1)$ is always present in the denominator for $n>2$, what means that when $p$ takes extreme values, $p=0$ or $p=1$, it is not possible to reach the global optimum from any solution, since the algorithm will keep the same solution if $p=0$ or will alternate between two solutions if $p=1$. However, when $n=1$ the probability $p=1$ is valid, furthermore, is optimal, because if the global solution is not present at the beginning we can reach it by alternating the only bit we have. In Figure~\ref{fig:onemax} we show the expected runtime as a function of the probability of flipping a bit for $n=1$ to 7. We can observe how the optimal probability (the one obtaining the minimum expected runtime) decreases as $n$ increases.

\begin{figure}[!ht]
\centering
\includegraphics[width=0.95\textwidth]{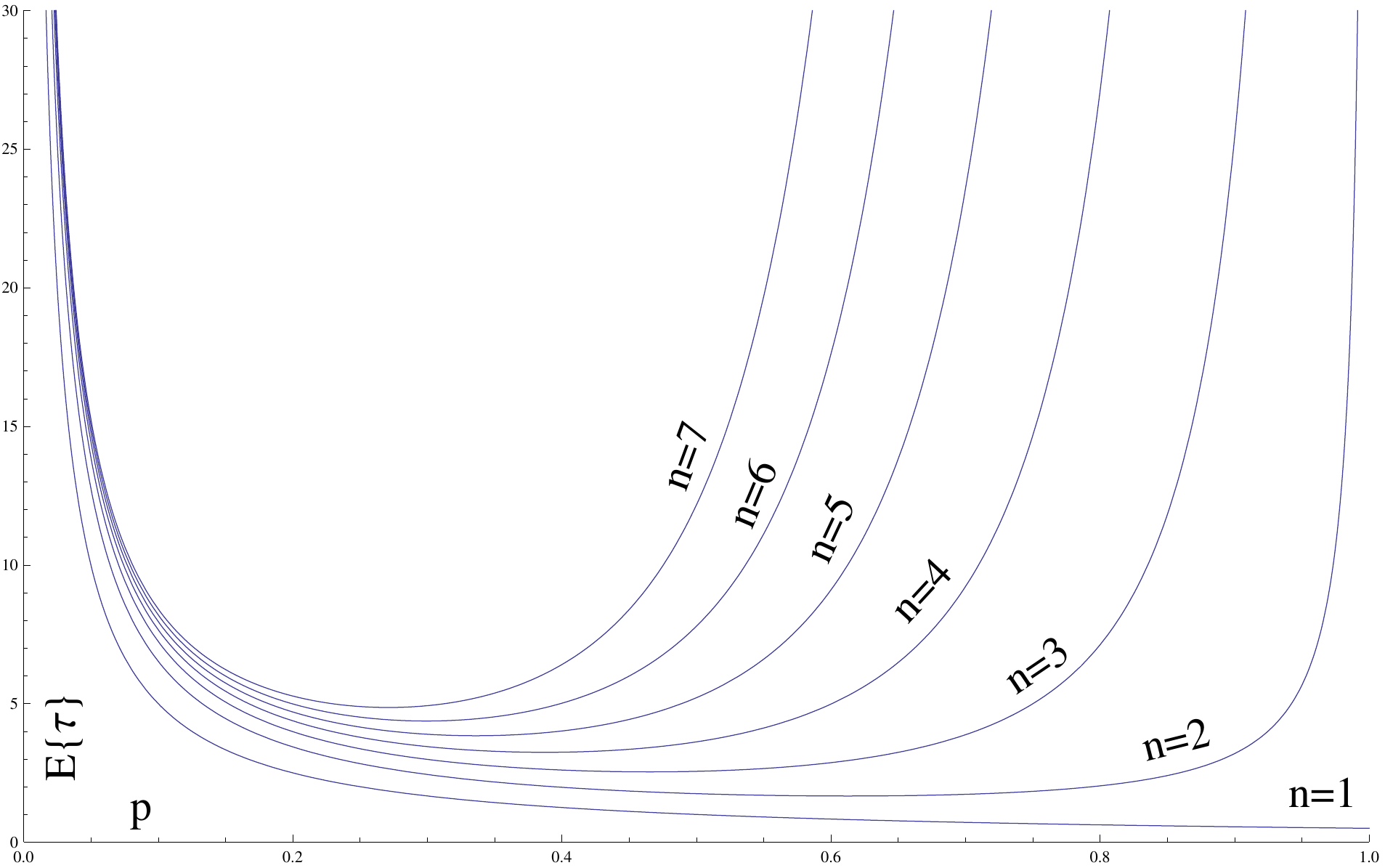}
\caption{Expected runtime of the $(1+1)$ EA for Onemax as a function of the probability of flipping a bit. Each line correspond to a different value of $n$ from 1 to 7.}
\label{fig:onemax}
\end{figure}

Having the exact expressions we can compute the optimal mutation probability for each $n$ by using classical optimization methods in one variable. In particular, for $n=1$ the optimal value is $p=1$ as we previously saw and for $n=2$ we have to solve a cubic polynomial in order to obtain the exact expression. The result is:
\begin{equation}
p_2^* = \frac{1}{5} \left(6-
\sqrt[3]{\frac{2}{23-5 \sqrt{21}}}
-\sqrt[3]{\frac{23-5 \sqrt{21}}{2}}
\right) \approx 0.561215 ,
\end{equation}
which is slightly higher than the recommended value $p=1/n$. Observe, however, that this result does not contradict the ones in~\cite{Witt2013tight}, since Witt's work provides asymptotic expressions and discards low-order terms, while we are working here with expressions for low $n$ values. In accordance to Witt's work, we would expect the optimal probability of mutation $p_n^*$ to approximate the $1/n$ value, and this is what happens.
As we increase $n$, analytical responses for the optimal probability are not possible and we have to apply numerical methods. In our case we used the Newton method in order to find a root of the equation $\frac{d \Exp{\tau}}{d p}=0$. Some results up to $n=100$ can be found in Table~\ref{tab:optimal-p}. A fast observation of the results reveals that the optimal probability is always a little bit higher than the recommended $p=1/n$.

\begin{table}[!ht]
\begin{center}
\begin{tabular}{|cr@{.}lr@{.}l||cr@{.}lr@{.}l|}
\hline
$n$ 
& \multicolumn{2}{c}{$p_n^*$} 
& \multicolumn{2}{c||}{$\Exp{\tau}$} 
& $n$ 
& \multicolumn{2}{c}{$p_n^*$} 
& \multicolumn{2}{c|}{$\Exp{\tau}$}  \\
\hline
1 & 1&00000 & 0&500 & 20 & 0&06133 & 127&453 \\
2 & 0&56122 & 2&959 & 30 & 0&04046 & 222&079 \\
3 & 0&38585 & 6&488 & 40 & 0&03009 & 325&900 \\
4 & 0&29700 & 10&808 & 50 & 0&02391 & 436&580 \\
5 & 0&24147 & 15&758 & 60 & 0&01981 & 552&734 \\
6 & 0&20323 & 21&222 & 70 & 0&01690 & 673&445 \\
7 & 0&17526 & 27&120 & 80 & 0&01473 & 798&059 \\
8 & 0&15391 & 33&391 & 90 & 0&01304 & 926&088 \\
9 & 0&13710 & 39&990 & 100 & 0&01170 & 1057&151 \\
10 & 0&12352 & 46&882 &    &  \multicolumn{2}{c}{~} &      \multicolumn{2}{c|}{~}         \\
\hline 
\end{tabular}
\end{center}
\caption{Optimal probability values for an $(1+1)$ EA solving Onemax.}
\label{tab:optimal-p}
\end{table}

Before we go further, we could question the use of an approximate method (the Newton method) over an exact expression to find the optimal probability for mutation. In particular, as we get approximate values anyway, why not directly run the $(1+\lambda)$ EA enough times to get an accurate enough approximated value for the optimal probability? Although in both cases we end with an approximated value, the approximation is done at a different level and using the Newton method, we get higher accuracy in less time. Let us explain this in detail. First, we have to say that given a probability of bit-flip $p$ the computation of the first-hitting time using~(\ref{eqn:expected-1+l}) is very fast and the result we obtain is exact up to the machine precision. If we want to obtain the expected first-hitting time running the algorithm we need to run it several thousand times to get a good confidence interval and the final approximation will be coarser than using the exact formulas. Just as an illustration we run the algorithm 1,000 times for $n=100$, $\lambda=1$ and $p=0.01$, and it took 193 s to find an expected runtime of $1058.60\pm21.73$ with 95\% confidence. On the other hand the evaluation of the exact expression was done in $0.837$ s and found an expected runtime of $1069.54$ with 11 decimals precision\footnote{The experiments were done in a MacBook Pro with an Intel Core i7 processor running at 2.8~GHz and 4~GB of DDR3 RAM memory.}.
Second, if we want to obtain the optimal probability for mutation we have to apply a numerical method. Using the exact formulas, we applied the Newton method and after several steps we obtain an approximated optimal probability $p^*$. In order to get an optimal mutation probability using the completely empirical approach we need to apply again a numerical method like the false position method (the Newton method cannot be applied now because we cannot compute derivatives). In order to evaluate each probability $p$ we need to run the algorithm thousands of times and, in the end we can only obtain an approximated value for the expected runtime. As an illustration we can say that the execution of the Newton method for $n=100$ stopped after $7.7$~s, which is the time required to run the $(1+1)$ EA algorithm around $40$ times in our machine. However, with 40 independent runs the precision of the expected runtime is very low. In summary, the completely empirical method requires much more time than the Newton method applied to the exact formulas for the same precision.

From previous work we know that the optimal probability is in the form $c/n$ for a constant $c$. We can use the results obtained by numerical analysis to find the value of $c$ and check the dependency with $n$. That is, using the optimal probability $p_n^*$ shown in Table~\ref{tab:optimal-p} we can compute $c_n = p_n^* n$ in order to see what is the value of $c_n$. In Figure~\ref{fig:constant} we plot $c_n$ as a function of $n$. We can observe that the optimal probability is not $p=c/n$ for a fixed $c$. The value of the constant $c_n$ is higher than 1 and depends on $n$. However, we can observe a clear trend $c_n \rightarrow 1$ as $n$ tends to $\infty$. The maximum value for $c_n$ is reached in $n=11$ and the value is $c_{11}=1.23559$.

\begin{figure}[!ht]
\centering
\includegraphics[width=0.95\textwidth]{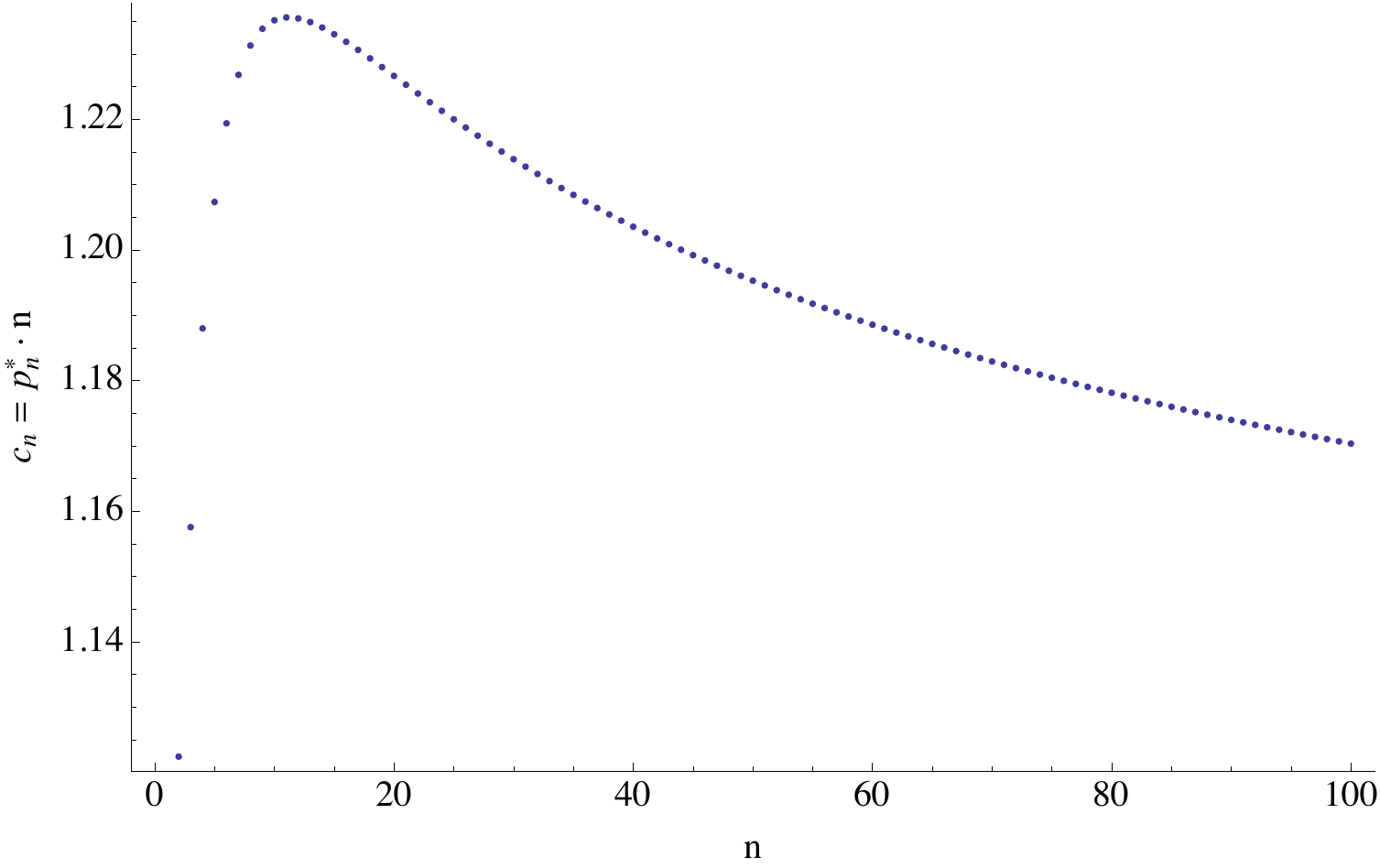}
\caption{The value of the constant $c$ in the optimal probability $p=c/n$ for Onemax instances from $2$ to $100$.}
\label{fig:constant}
\end{figure}

It is also well-known that for this optimal probability the expected runtime is $\Theta (n \log n)$. We can also check this using numerical analysis. We used the optimal expected runtime for this computation and found the best fit model including $n$ and $n \log n$ terms. The result is:
\begin{equation}
\Exp{\tau} \approx  -1.51165 n + 2.62161 n \log n ,
\end{equation}
where we can observe how the factor in front of the $n \log n$ term is near $e$, which is the theoretically predicted factor for $c=1$ and large values of $n$~\citep{Witt2013tight}.

Let us now study the expected runtime for different values of $\lambda$. In this case we fix the size of the problem to $n=50$ and we analyze the expected runtime using $p=1/n$ for $\lambda=1$ to 50. The results are shown in Figure~\ref{fig:lambda} in both, natural scale and log-log scale.

\begin{figure}[!ht]
\centering
\subfigure[Expected runtime as a function of $\lambda$]{
\includegraphics[width=0.47\textwidth]{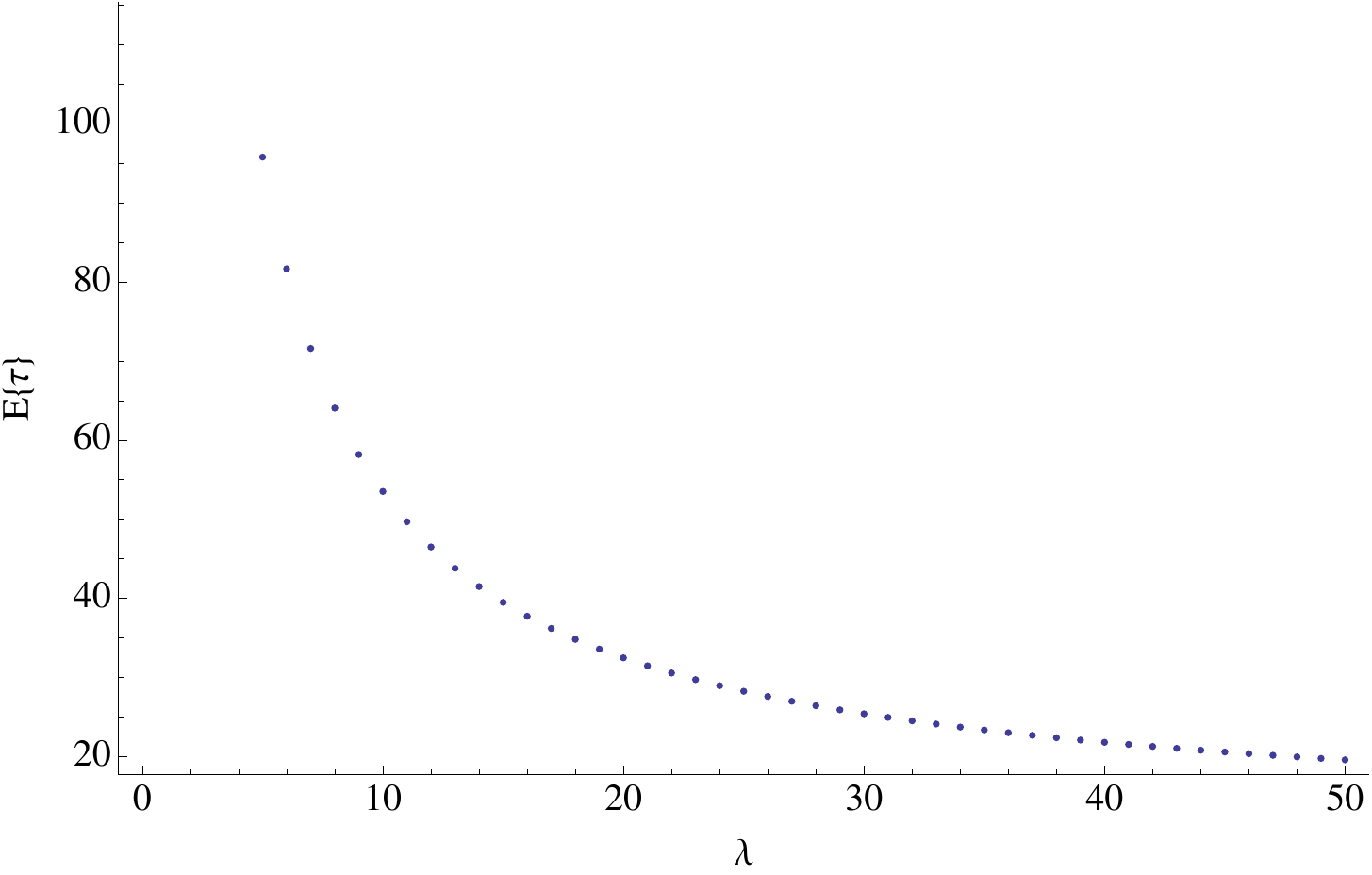}}
\subfigure[Log-log plot of the expected runtime against $\lambda$]{
\includegraphics[width=0.47\textwidth]{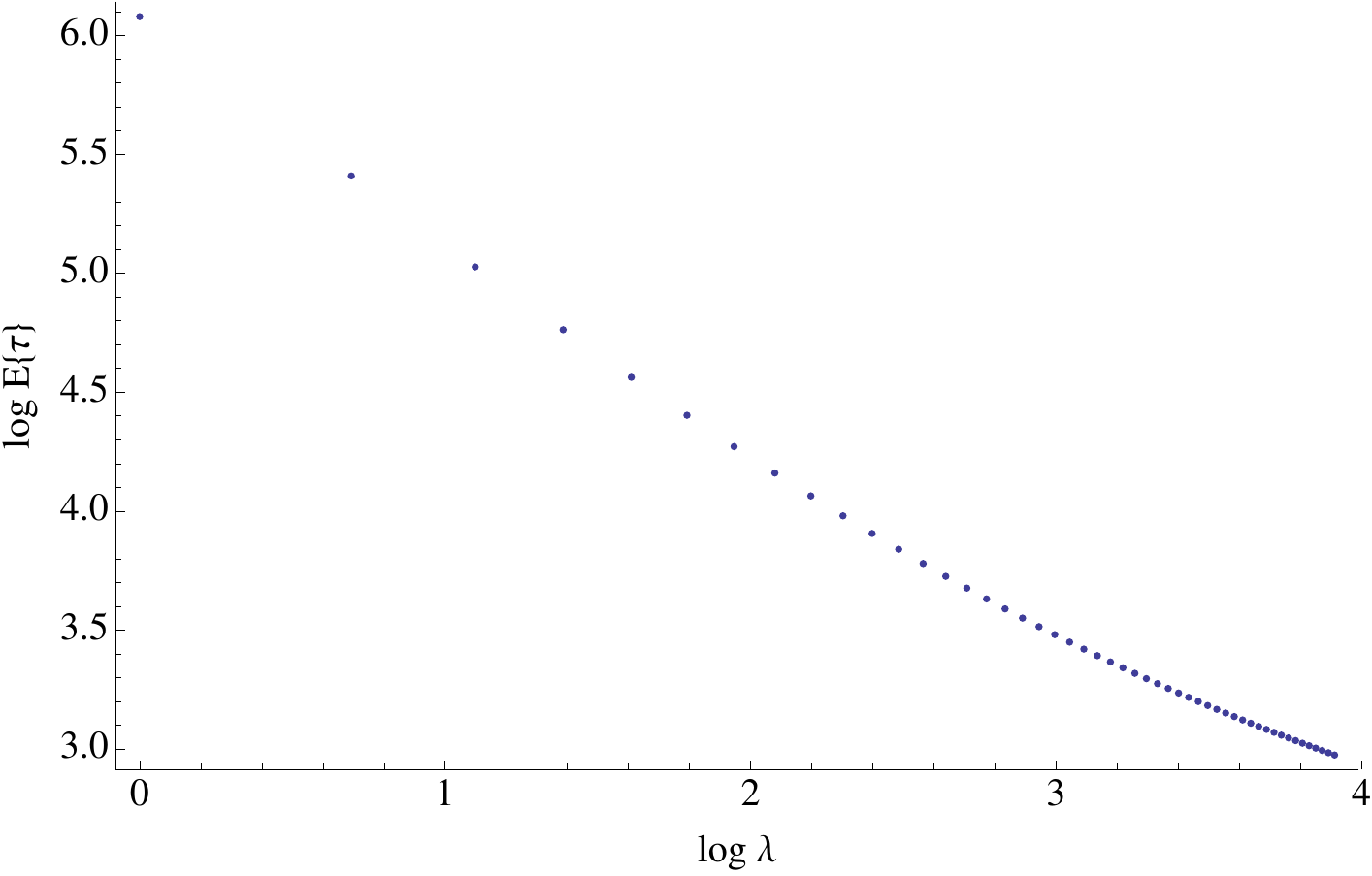}}
\caption{The expected runtime of a $(1+1)$ EA with $p=1/n$ and $n=50$ for $\lambda=1$ to 50.}
\label{fig:lambda}
\end{figure}

The log-log plot is almost a straight line, what suggests that we can express the expected runtime as a potential function of $\lambda$. The best linear regression model for the log-log plot is:
\begin{equation}
\log \Exp{\tau} \approx 5.78452 - 0.746412 \log \lambda ,
\end{equation}
or equivalently:
\begin{equation}
\Exp{\tau} \approx \frac{325.226}{\lambda^{0.746412}} .
\end{equation}

However, we cannot compare this model with the result of \cite{Jansen2005} of $O(n \log n/\lambda +n)$. Thus, we found the best fit model in the form $\Exp{\tau} = A + \frac{B}{\lambda}$ and we got:
\begin{equation}
\Exp{\tau} \approx 11.1306 +\frac{424.99}{\lambda } .
\end{equation}

The value of the constant $A$ is small enough to say that the expected runtime is approximately divided by $\lambda$ when we generate $\lambda$ offspring.

%

\section{Conclusions}
\label{sec:conclusions}

We analyzed the bit-flip mutation operator from the point of view of landscape theory. In particular, we derived closed-form formulas for all the statistical moments of the fitness distribution of a mutated solution. These moments can be expressed as a polynomial in $p$, the probability of flipping a bit. Using the moments we derived an expression for the probability mass function of the fitness value after applying bit-flip mutation to a given solution. The expression takes an elegant matrix form in which we can distinguish a problem-dependent part and an operator-dependent part. The problem-dependent part can be obtained using the elementary landscape decomposition of the objective function of the problem and their powers. The operator-dependent part depends only on the probability $p$.

We also derived the problem-dependent part for two well-known problems: Onemax and MAX-SAT. In the first case, the problem-dependent part is especially simple and efficient to compute. This allowed us to derive the exact expression for the runtime of an $(1+\lambda)$ EA for solving Onemax, finding a connection between landscape theory and runtime analysis. Using this expression we obtained the optimal probability for bit-flip mutation as a function of $n$, the number of bits.

It is possible to analyze other operators in the same way we did with bit-flip mutation. Thus, we think that an interesting future line of research could be the application of similar ideas to find the probability mass function of the distribution after the application of several chained operators. In particular, recent developments in landscape theory suggest that it is possible to analyze the fitness distribution of the offspring of two parent solutions when the uniform crossover is applied~\citep{Chicano2012crossover}. These results together with the connection between landscape theory and runtime analysis shown in this paper could provide a natural way of introducing crossover in the runtime results.

\section{Acknowledgements}

This work has been partially funded by the Spanish Ministry of Economy and Competitiveness and FEDER under contract TIN2011-28194 (the
roadME project), and by the Air Force Office of
Scientific Research, Air Force Materiel Command, USAF, under grant
number FA9550-08-1-0422. The U.S. Government is authorized to
reproduce and distribute reprints for Governmental purposes
notwithstanding any copyright notation thereon.

The authors would also like to thank the organizers and participants
of the seminars on Theory of Evolutionary Algorithms (10361 and 13271) at
Schlo\ss\ Dagstuhl - Leibniz-Zentrum f\"ur Informatik.

\bibliographystyle{apalike}
\bibliography{landscapes}

\begin{thebibliography}{}

\bibitem[Biyikoglu et~al., 2007]{Biyikoglu2007}
Biyikoglu, T., Leyold, J., and Stadler, P.~F. (2007).
\newblock {\em Laplacian Eigenvectors of Graphs}.
\newblock Lecture Notes in Mathematics. Springer-Verlag.

\bibitem[Chicano and Alba, 2011]{Chicano2011gecco}
Chicano, F. and Alba, E. (2011).
\newblock Exact computation of the expectation curves of the bit-flip mutation
  using landscapes theory.
\newblock In {\em Proceedings of the 13th Annual Conference on Genetic and
  Evolutionary Computation}, pages 2027--2034. ACM.

\bibitem[Chicano et~al., 2012]{Chicano2012crossover}
Chicano, F., Whitley, D., and Alba, E. (2012).
\newblock Exact computation of the expectation curves for uniform crossover.
\newblock In {\em Proceedings of the 14th annual Conference on Genetic and
  Evolutionary Computation}, pages 1301--1308. ACM.

\bibitem[Chicano et~al., 2011]{Chicano2011ecj}
Chicano, F., Whitley, L.~D., and Alba, E. (2011).
\newblock A methodology to find the elementary landscape decomposition of
  combinatorial optimization problems.
\newblock {\em Evolutionary Computation}, 19(4):597--637.

\bibitem[Feinsilver and Kocik, 2005]{Feinsilver2005}
Feinsilver, P. and Kocik, J. (2005).
\newblock Krawtchouk polynomials and krawtchouk matrices.
\newblock In Baeza-Yates, R., Glaz, J., Gzyl, H., H\"usler, J., and Palacios,
  J., editors, {\em Recent Advances in Applied Probability}, pages 115--141.
  Springer US.

\bibitem[Garnier et~al., 1999]{Garnier1999}
Garnier, J., Kallel, L., and Schoenauer, M. (1999).
\newblock Rigorous hitting times for binary mutations.
\newblock {\em Evolutionary Computation}, 7(2):173--203.

\bibitem[Grover, 1992]{Grover1992local}
Grover, L.~K. (1992).
\newblock Local search and the local structure of {NP}-complete problems.
\newblock {\em Operations Research Letters}, 12:235--243.

\bibitem[He and Yao, 2003]{He:Yao2003}
He, J. and Yao, X. (2003).
\newblock Towards an analytic framework for analysing the computation time of
  evolutionary algorithms.
\newblock {\em Artificial Intelligence}, 145:59 -- 97.

\bibitem[Iosifescu, 1980]{Iosifescu1980}
Iosifescu, M. (1980).
\newblock {\em Finite Markov Processes and Their Applications}.
\newblock John Wiley \& Sons.

\bibitem[Jansen et~al., 2005]{Jansen2005}
Jansen, T., {De Jong}, K.~A., and Wegener, I. (2005).
\newblock On the choice of the offspring population size in evolutionary
  algorithms.
\newblock {\em Evolutionary Computation}, 13(4):413--440.

\bibitem[Mirsky, 1955]{Mirsky1955}
Mirsky, L. (1955).
\newblock {\em An Introduction to Linear Algebra}.
\newblock Clarendon Press.

\bibitem[Neumann and Witt, 2010]{Neumann2010}
Neumann, F. and Witt, C. (2010).
\newblock {\em Bioinspired Computation in Combinatorial Optimization}.
\newblock Springer-Verlag.

\bibitem[Reidys and Stadler, 2002]{Reidys2002}
Reidys, C.~M. and Stadler, P.~F. (2002).
\newblock Combinatorial landscapes.
\newblock {\em SIAM Review}, 44(1):3--54.

\bibitem[Stadler, 1995]{Stadler1995landscapes}
Stadler, P.~F. (1995).
\newblock Toward a theory of landscapes.
\newblock In L{\'o}pez-Pe{\~n}a, R., Capovilla, R., Garc{\'\i}a-Pelayo, R.,
  H.Waelbroeck, and Zertruche, F., editors, {\em Complex Systems and Binary
  Networks}, pages 77--163. Springer-Verlag.

\bibitem[Sutton et~al., 2010]{Sutton2010}
Sutton, A.~M., Howe, A.~E., and Whitley, L.~D. (2010).
\newblock Directed plateau search for {MAX}-k-{SAT}.
\newblock In {\em Proceedings of the 3rd Annual Symposium on Combinatorial
  Search}, pages 90--97.

\bibitem[Sutton et~al., 2011a]{Sutton2011foga}
Sutton, A.~M., Whitley, D., and Howe, A.~E. (2011a).
\newblock Approximating the distribution of fitness over hamming regions.
\newblock In {\em Proceedings of the 11th Workshop Proceedings on Foundations
  of Genetic Algorithms}, pages 93--104. ACM.

\bibitem[Sutton et~al., 2011b]{Sutton2011gecco}
Sutton, A.~M., Whitley, D., and Howe, A.~E. (2011b).
\newblock Mutation rates of the (1+1)-{EA} on pseudo-boolean functions of
  bounded epistasis.
\newblock In {\em Proceedings of the 13th Annual Conference on Genetic and
  Evolutionary Computation}, pages 973--980. ACM.

\bibitem[Sutton et~al., 2009]{Sutton2009}
Sutton, A.~M., Whitley, L.~D., and Howe, A.~E. (2009).
\newblock A polynomial time computation of the exact correlation structure of
  k-satisfiability landscapes.
\newblock In {\em Proceedings of the 11th Annual Conference on Genetic and
  Evolutionary Computation}, pages 365--372. ACM.

\bibitem[Sutton et~al., 2011c]{Sutton2011tcs}
Sutton, A.~M., Whitley, L.~D., and Howe, A.~E. (2011c).
\newblock Computing the moments of k-bounded pseudo-boolean functions over
  hamming spheres of arbitrary radius in polynomial time.
\newblock {\em Theoretical Computer Science}, 425:58--74.

\bibitem[Terras, 1999]{Terras1999}
Terras, A. (1999).
\newblock {\em Fourier Analysis on Finite Groups and Applications, Cambridge U.
  Press, Cambridge}.
\newblock Cambridge University Press.

\bibitem[Vose, 1999]{Vose1999}
Vose, M.~D. (1999).
\newblock {\em The Simple Genetic Algorithm: Foundations and Theory}.
\newblock MIT Press.

\bibitem[Walsh, 1923]{Walsh1923}
Walsh, J.~L. (1923).
\newblock A closed set of normal orthogonal functions.
\newblock {\em American Journal of Mathematics}, 45(1):5--24.

\bibitem[Whitley et~al., 2008]{Whitley2008}
Whitley, D., Sutton, A.~M., and Howe, A.~E. (2008).
\newblock Understanding elementary landscapes.
\newblock In {\em Proceedings of the 10th Annual Conference on Genetic and
  Evolutionary Computation}, pages 585--592, New York, NY, USA. ACM.

\bibitem[Whitley and Sutton, 2009]{WhitleySutton2009}
Whitley, L.~D. and Sutton, A.~M. (2009).
\newblock Partial neighborhoods of elementary landscapes.
\newblock In {\em Proceedings of the 11th Annual Conference on Genetic and
  Evolutionary Computation}, pages 381--388. ACM.

\bibitem[Witt, 2013]{Witt2013tight}
Witt, C. (2013).
\newblock Tight bounds on the optimization time of a randomized search
  heuristic on linear functions.
\newblock {\em Combinatorics, Probability and Computing}, 22(2):298--314.

\end{thebibliography}
\end{document}